\documentclass{llncs}

\usepackage{amsmath,amssymb
}
\usepackage{xspace}
\usepackage{dsfont}
\usepackage[colorlinks=true,urlcolor=webblue,linkcolor=webgreen,filecolor=webblue,citecolor=webgreen,pdfpagemode=UseOutlines,pdfstartview=FitH,pdfpagelayout=OneColumn,bookmarks=true]{hyperref}
\usepackage{enumerate}
\usepackage{graphicx}
\usepackage{enumitem}
\usepackage{xcolor}
\usepackage{authblk}
\usepackage{fullpage}
\usepackage{cleveref}

\definecolor{mgreen}{rgb}{.1,.7,0}

\usepackage{algorithm, algpseudocode, float}

\usepackage[n,keys,asymptotics]{cryptocode}

\newtheorem{construction}{Construction}

\newcommand{\Z}{\mathbb{Z}}
\newcommand{\N}{\mathbb{N}}

\newcommand{\C}{\mathbb{C}}
\newcommand{\ket}[1]{|#1\rangle} %
\newcommand{\bra}[1]{\langle #1|} %
\newcommand{\tr}[1]{\mathrm{Tr}\left(#1\right)} %

\newcommand\id{\mathbb{I}}

\newcommand{\one}{\mathds 1}

\definecolor{webgreen}{rgb}{0,.5,0}
\definecolor{webblue}{rgb}{0,0,.5}
\newcommand\algo{\mathcal}

\newcommand{\Ver}{\ensuremath{\mathsf{Ver}}\xspace}

\newcommand{\Mint}{\ensuremath{\mathsf{Mint}}\xspace}

\newcommand{\Gen}{\ensuremath{\mathsf{Gen}}\xspace}
\newcommand{\Init}{\ensuremath{\mathsf{Init}}\xspace}
\newcommand{\Reflect}{\ensuremath{\mathsf{Reflect}}\xspace}
\newcommand{\CReflect}{\ensuremath{\mathsf{CReflect}}\xspace}
\newcommand{\Eval}{\ensuremath{\mathsf{Eval}}\xspace}
\newcommand{\Invert}{\ensuremath{\mathsf{Invert}}\xspace}

\newcommand{\expref}[2]{\texorpdfstring{\hyperref[#2]{#1~\ref{#2}}}{#1~\ref{#2}}}
\newcommand{\bits}{\{0,1\}}
\newcommand{\bit}{\{0,1\}}

\newcommand{\from}{\leftarrow}

\newcommand{\acc}{\ensuremath{\mathsf{acc}}}
\newcommand{\rej}{\ensuremath{\mathsf{rej}}}

\newcommand{\supp}{\ensuremath{\mathsf{supp}}}

\newcommand{\snt}{\ensuremath{S^{\scriptscriptstyle \uparrow}}}

\newcommand{\idealstate}{\ensuremath{\mathfrak{IS}}\xspace}
\newcommand{\effstate}{\ensuremath{\mathfrak{ES}}\xspace}
\newcommand{\idealunitary}{\ensuremath{\mathfrak{IU}}\xspace}
\newcommand{\effunitary}{\ensuremath{\mathfrak{EU}}\xspace}
\newcommand{\haarmoney}{\ensuremath{\mathfrak{HM}}\xspace}

\newcommand{\Untrace}{\ensuremath{\mathsf{Untrace}}}

\renewcommand{\negl}{\ensuremath{\mathsf{negl}}\xspace}

\newcommand{\proj}[1]{\ensuremath{|#1\rangle \langle #1|}}

\newcommand{\ketbra}[2]{\left|#1\right\rangle\!\!\left\langle #2\right|}
\newcommand{\braket}[2]{\left\langle #1 \mid #2 \right\rangle}
\newcommand{\Hi}{\mathcal{H}}
\newcommand{\Tr}{\mathrm{Tr}}

\newcommand{\microspace}{\mspace{.5mu}} %
\renewcommand{\ket}[1]{\ensuremath{\lvert\microspace #1
		\microspace\rangle}} %
\renewcommand{\bra}[1]{\ensuremath{\langle\microspace #1
		\microspace\rvert}} %

\newif\ifsubmission
\submissionfalse
\ifsubmission
\newcommand{\ga}[1]{}
\newcommand{\cm}[1]{}
\else
\newcommand{\ga}[1]{{\noindent \textcolor{purple}{\emph{(GA:  #1)}}}{}}
\newcommand{\cm}[1]{{\noindent \textcolor{mgreen}{\emph{(CM:  #1)}}}{}}
\fi

\newcommand{\Sym}{\ensuremath{\mathrm{Sym}}\xspace}
\newcommand{\BSym}{\ensuremath{\mathbf{Sym}}\xspace}

\title{Efficient simulation of random states and random unitaries}
\date{\today}

\author{Gorjan Alagic\inst{1,2} \and Christian Majenz\inst{3,4} \and Alexander Russell\inst{5}}

\institute{Joint Center for Quantum Information and Computer Science, University of Maryland\and National Institute of Standards and Technology, Gaithersburg, Maryland\and  QuSoft, Amsterdam\and Centrum Wiskunde \& Informatica, Amsterdam\and Department of Computer Science and Engineering, University of Connecticut}

\makeindex

\begin{document}

\maketitle

\begin{abstract}
We consider the problem of efficiently simulating random quantum states and random unitary operators, in a manner which is convincing to unbounded adversaries with black-box oracle access. 

~~~This problem has previously only been considered for restricted adversaries. Against adversaries with an a priori bound on the number of queries, it is well-known that $t$-designs suffice. Against polynomial-time adversaries, one can use pseudorandom states (PRS) and pseudorandom unitaries (PRU), as defined in a recent work of Ji, Liu, and Song; unfortunately, no provably secure construction is known for PRUs.

~~~In our setting, we are concerned with unbounded adversaries. Nonetheless, we are able to give stateful quantum algorithms which simulate the ideal object in both settings of interest. In the case of Haar-random states, our simulator is polynomial-time, has negligible error, and can also simulate verification and reflection through the simulated state. This yields an immediate application to quantum money: a money scheme which is information-theoretically unforgeable and untraceable. In the case of Haar-random unitaries, our simulator takes polynomial space, but simulates both forward and inverse access with zero error. 

~~~These results can be seen as the first significant steps in developing a theory of lazy sampling for random quantum objects.
\end{abstract}


\section{Introduction}

\subsection{Motivation} 

Efficient simulation of randomness is a task with countless applications, ranging from cryptography to derandomization. In the setting of classical probabilistic computation, such simulation is straightforward in many settings. For example, a random function which will only be queried an a priori bounded number of times $t$ can be perfectly simulated using a $t$-wise independent function~\cite{CW81}. In the case of unbounded queries, one can use pseudorandom functions (PRFs), provided the queries are made by a polynomial-time algorithm~\cite{GGM86}. These are examples of \emph{stateless} simulation methods, in the sense that the internal memory of the simulator is initialized once (e.g., with the PRF key) and then remains fixed regardless of how the simulator is queried. Against arbitrary adversaries, one must typically pass to \emph{stateful} simulation. For example, the straightforward and well-known technique of \emph{lazy sampling} suffices to perfectly simulate a random function against arbitrary adversaries; however, the simulator must maintain a list of responses to all previous queries. 

Each of these techniques for simulating random classical primitives has a plethora of applications in theoretical cryptography, both as a proof tool and for cryptographic constructions. These range from constructing secure cryptosystems for encryption and authentication, to proving security reductions in a wide range of settings, to establishing security in idealized models such as the Random Oracle Model~\cite{BR93}.

\subsubsection{Quantum randomness.}

As is well-known, quantum sources of randomness exhibit dramatically different properties from their classical counterparts~\cite{MY04,BB84}. Compare, for example, uniformly random $n$-bit classical states (i.e., $n$-bit strings) and uniformly random $n$-qubit (pure) quantum states. A random string $x$ is obviously trivial to sample perfectly given probabilistic classical (or quantum) computation, and can be copied and distributed arbitrarily. However, it is also (just as obviously) deterministic to all parties who have examined it before. By contrast, a random state $\ket{\varphi}$ would take an unbounded amount of information to describe perfectly. Even if one manages to procure such a state, it is then impossible to copy due to the no-cloning theorem. On the other hand, parties who have examined $\ket{\varphi}$ many times before, can \emph{still} extract almost exactly $n$ bits of randomness from any fresh copy of $\ket{\varphi}$ they receive -- even if they use the \emph{exact same measurement procedure} each time.

The differences between random classical and random quantum maps are even more stark. The outputs of a classical random function are of course classical random strings, with all of the aforementioned properties. Outputs which have already been examined become effectively deterministic, while the rest remain uniformly random and independent. This is precisely what makes efficient simulation possible via lazy sampling. A Haar-random unitary $U$ queried on two inputs $\ket{\psi}$ and $\ket{\phi}$ also produces (almost) independent and uniformly random states when queried, but only if the queries are \emph{orthogonal}, i.e., $\braket{\psi}{\phi} = 0$. Unitarity implies that overlapping queries must be answered consistently, i.e., if $\braket{\psi}{\phi} = \delta$ then $\braket{(U\psi)}{(U\phi)} = \delta$. This possibility of querying with a distinct pure state which is not linearly independent from previous queries simply doesn't exist for classical functions. 

We emphasize that the above differences should not be interpreted as quantum random objects simply being ``stronger'' than their classical counterparts. In the case of classical states, i.e. strings, the ability to copy is quite useful, e.g., in setting down basic security definitions~\cite{BZ13a,AGM18,Alagic2018a} or when rewinding an algorithm~\cite{Unruh2012,Watrous2009,Don2019}. In the case of maps, determinism is also quite useful, e.g., for verification in message authentication.

\subsection{The problem: efficient simulation}

Given the dramatic differences between classical and quantum randomness, and the usefulness of both, it is reasonable to ask if there exist quantum analogues of the aforementioned efficient simulators of classical random functions. In fact, given the discussion above, it is clear that we should begin by asking if there even exist efficient simulators of random quantum states.

\subsubsection{Simulating random states.}

The first problem of interest is thus to efficiently simulate the following ideal object: an oracle $\idealstate(n)$ which contains a description of a perfectly Haar-random $n$-qubit pure state $\ket{\varphi}$, and which outputs a copy of $\ket{\varphi}$ whenever it is invoked. We first make an obvious observation: the classical analogue, which is simply to generate a random bitstring $x \leftarrow \bit^n$ and then produce a copy whenever asked, is completely trivial. In the quantum case, efficient simulation is only known against limited query algorithms (henceforth, adversaries.) 

If the adversary has an a priori bound on the number of queries, then \emph{state $t$-designs} suffice. These are indexed families $\{\ket{\varphi_{k, t}} : k \in K_t\}$ of pure states which perfectly emulate the standard uniform ``Haar'' measure on pure states, up to the first $t$ moments. State $t$-designs can be sampled efficiently, and thus yield a stateless simulator for this case~\cite{AJ07}. A recent work of Ji, Liu and Song considered the case of polynomial-time adversaries~\cite{Ji2018}. They defined a notion of \emph{pseudorandom states} (PRS), which appear Haar-random to polynomial-time adversaries who are allowed as many copies of the state as they wish. They also showed how to construct PRS efficiently, thus yielding a stateless simulator for this class of constrained adversaries~\cite{Ji2018}; see also~\cite{Brakerski2019}. 

The case of arbitrary adversaries is, to our knowledge, completely unexplored. In particular, before this work it was not known whether simulating $\idealstate(n)$ against adversaries with no a priori bound on query or time complexity is possible, even if given polynomial space (in $n$ and the number of queries) and unlimited time. Note that, while the state family constructions from \cite{Ji2018,Brakerski2019} could be lifted to the unconditional security setting by instantiating them with random instead of pseudorandom functions, this would require space exponential in $n$ regardless of the number of queries.

\subsubsection{Simulating random unitaries.}

In the case of simulating random unitaries, the ideal object is an oracle \idealunitary(n) which contains a description of a perfectly Haar-random $n$-qubit unitary operator $U$, and applies $U$ to its input whenever it is invoked. The classical analogue is the well-known Random Oracle, and can be simulated perfectly using the aforementioned technique of lazy sampling. In the quantum case, the situation is even less well-understood than in the case of states.

For the case of query-limited adversaries, we can again rely on design techniques: (approximate) \emph{unitary $t$-designs} can be sampled efficiently, and suffice for the task~\cite{BHH16,L10}. Against polynomial-time adversaries, Ji, Liu and Song defined the natural notion of a \emph{pseudorandom unitary} (or PRU) and described candidate constructions~\cite{Ji2018}. Unfortunately, at this time there are no provably secure constructions of PRUs. As in the case of states, the case of arbitrary adversaries is completely unexplored. Moreover, one could a priori plausibly conjecture that simulating \idealunitary might even be impossible. The no-cloning property seems to rule out examining input states, which in turn seems to make it quite difficult for a simulator to correctly identify the overlap between multiple queries, and then answer correspondingly.

\subsubsection{Extensions.}

While the above problems already appear quite challenging, we mention several natural extensions that one might consider. First, for the case of repeatedly sampling a random state $\ket{\varphi}$, one would ideally want some additional features, such as the ability to apply the two-outcome measurement $\{\proj{\varphi}, \one - \proj{\varphi}\}$ (\emph{verification}) or the reflection $\one - 2\proj{\varphi}$. In the case of pseudorandom simulation, these additional features can be used to create a (computationally secure) quantum money scheme~\cite{Ji2018}. For the case of simulating random unitaries, we might naturally ask that the simulator for a unitary $U$ also has the ability to respond to queries to $U^{-1}=U^\dagger$.

\subsection{This work} 

In this work, we make significant progress on the above problems, by giving the first simulators for both random states and random unitaries, which are convincing to arbitrary adversaries. We also give an application of our sampling ideas: the construction of a new quantum money scheme, which provides information-theoretic security guarantees against both forging and tracing. 

We begin by remarking that our desired simulators must necessarily be stateful, for both states and unitaries. Indeed,  since approximate $t$-designs have $\Omega((2^{2n}/t)^{2t}) $ elements~(see, e.g., \cite{Roy2009} which provides a more fine-grained lower bound), a stateless approach would require superpolynomial space simply to store an index from a set of size  $\Omega((2^{2n}/t(n))^{2t(n)}) $ for all polynomials $t(n)$.

In the following, we give a high-level overview of our approach for each of the two simulation problems of interest.

\subsubsection{Simulating random states.}\label{sec:intro-state}

As discussed above, we wish to construct an efficient simulator $\effstate(n)$ for the ideal oracle $\idealstate(n)$. For now we focus on simulating the procedure which generates copies of the fixed Haar-random state; we call this $\idealstate(n).\Gen$. We first note that the mixed state observed by the adversary after $t$ queries to $\idealstate(n).\Gen$ is the expectation of the projector onto $t$ copies of $\ket{\psi}$. Equivalently, it is the (normalized) projector onto the \emph{symmetric subspace} $\BSym_{n, t}$ of $(\C^{2^n})^{\otimes t}$:
	\begin{equation}
		\tau_t = \mathbb E_{\psi\sim\mathrm{Haar}} \proj{\psi}^{\otimes t} \propto \Pi_{\Sym^t\C^{2^n}}\,.
	\end{equation}
Recall that $\BSym_{n, t}$ is the subspace of $(\C^{2^n})^{\otimes t}$ of vectors which are invariant under permutations of the $t$ tensor factors. Our goal will be to maintain an entangled state between the adversary $\algo A$ and our oracle simulator $\effstate$ such that the reduced state on the side of $\algo A$ is $\tau_t$ after $t$ queries. Specifically, the joint state will be the maximally entangled state between the $\BSym_{n, t}$ subspace of the $t$ query output registers received by $\algo A$, and the $\BSym_{n, t}$ subspace of $t$ registers held by $\effstate$. If we can maintain this for the first $t$ queries, then it's not hard to see that there exists an isometry $V^{t \rightarrow t+1}$ which, by acting only on the state of $\effstate$, implements the extension from the $t$-fold to the $(t+1)$-fold joint state.

The main technical obstacle, which we resolve, is showing that $V^{t \rightarrow t+1}$ can be performed efficiently. To achieve this, we develop some new algorithmic tools for working with symmetric subspaces, including an algorithm for coherent preparation of its basis states. We let $A$ denote an $n$-qubit register, $A_j$ its indexed copies, and $A^t = A_1 \cdots A_t$ $t$-many indexed copies (and likewise for $B$.) We also let $\{\ket{\Sym(\alpha)} : \alpha \in \snt_{n,t}\}$ denote a particular orthonormal basis set for $\BSym_{n, t}$, indexed by some set $\snt_{n,t}$ (see \expref{Section}{sec:state} for definitions of these objects.)

\begin{theorem}\label{thm:sym-intro}
For each $n$ and $t$, there exists a polynomial-time quantum algorithm which implements an isometry $V=V^{t \to t+1}$ from $B^t$ to $A_{t+1}B^{t+1}$ such that, up to negligible trace distance,
\[
(\one_{A^t} \otimes V) \sum_{\alpha \in {\snt_{n,t}}} \ket{\Sym(\alpha)}_{A^t}\ket{\Sym(\alpha)}_{B^t}=\!\!\!\sum_{\beta \in {\snt_{n,t+1}}} \ket{\Sym(\beta)}_{A^{t+1}}\ket{\Sym(\beta)}_{B^{t+1}}\,.
\]
\end{theorem}

Above, $V$ is an operator defined to
apply to a specific subset of registers of a state. When no confusion
can arise, in such settings we will abbreviate $\one \otimes V$---the
application of this operator on the entire state---as simply $V$.

It will be helpful to view $V^{t \to t+1}$ as first preparing $\ket{0^n}_{A_{t+1}}\ket{0^n}_{B_{t+1}}$ and then applying a
unitary $U^{t \to t+1}$ on $A_{t+1}B^{t+1}$. \expref{Theorem}{thm:sym-intro} then gives us a way to answer $\Gen$ queries efficiently, as follows. For the first query, we prepare a maximally entangled state $\ket{\phi^+}_{A_1B_1}$ across two $n$-qubit registers $A_1$ and $B_1$, and reply with register $A_1$. Note that $\BSym_{n, 1} = \C^{2^n}$. For the second query, we prepare two fresh registers $A_2$ and $B_2$, both in the $\ket{0^n}$ state, apply $U^{1 \to 2}$ on $A_2B_1B_2$, return $A_2$, and keep $B_1B_2$. For the $t$-th query, we proceed similarly, preparing fresh blank registers $A_{t+1}B_{t+1}$, applying $U^{t \to t+1}$, and then outputting the register $A_{t+1}$.

With this approach, as it turns out, there is also a natural way to respond to verification queries $\Ver$ and reflection queries $\Reflect$. The ideal functionality \idealstate.\Ver is to apply the two-outcome measurement $\{\proj{\varphi}, \one - \proj{\varphi}\}$ corresponding to the Haar-random state $\ket{\varphi}$. To simulate this after producing $t$ samples, we apply the inverse of $U^{t-1 \to t}$, apply the measurement $\{\proj{0^{2n}}, \one - \proj{0^{2n}}\}$ to $A_tB_t$, reapply $U^{t-1 \to t}$, and then return $A_t$ together with the measurement outcome (i.e., yes/no).  For \idealstate.\Reflect, the ideal functionality is to apply the reflection $\one - 2\proj{\varphi}$ through the state. To simulate this, we perform a sequence of operations analogous to \Ver, but apply a phase of $-1$ on the $\ket{0^{2n}}$ state of $A_tB_t$ instead of measuring.

Our main result on simulating random states is to establish that this collection of algorithms correctly simulates the ideal object $\idealstate$, in the following sense.

\begin{theorem}\label{thm:state-sampler-summary}
There exists a stateful quantum algorithm $\effstate(n, \epsilon)$ which runs in time polynomial in $n$, $\log(1/\epsilon)$, and the number of queries $q$ submitted to it, and satisfies the following. For all oracle algorithms $\algo A$,
$$
\left|\Pr\left[\algo A^{\idealstate(n)} = 1\right] - \Pr\left[\algo A^{\effstate(n, \epsilon)} = 1\right]\right| \leq \epsilon\,.
$$
\end{theorem}


A complete description of our construction, together with the proofs of \expref{Theorem}{thm:sym-intro} and \expref{Theorem}{thm:state-sampler-summary}, are given in \expref{Section}{sec:state}. 

\subsubsection{Application: untraceable quantum money.}

To see that the efficient state sampler leads to a powerful quantum money scheme, consider building a scheme where the bank holds the ideal object $\idealstate.$ The bank can mint bills by $\idealstate.\Gen$, and verify them using $\idealstate.\Ver$. As each bill is guaranteed to be an identical and Haar-random state, it is clear that this scheme should satisfy perfect unforgeability and untraceability, under quite strong notions of security. 

By \expref{Theorem}{thm:state-sampler}, the same properties should carry over for a money scheme built on $\effstate$, provided $\epsilon$ is sufficiently small. We call the resulting scheme \emph{Haar money}. Haar money is an information-theoretically secure analogue of the scheme of~\cite{Ji2018}, which is based on pseudorandom states. We remark that our scheme requires the bank to have quantum memory and to perform quantum communication with the customers. However, given that quantum money already requires customers to have large-scale, high-fidelity quantum storage, these additional requirements seem reasonable.

The notions of correctness and unforgeability (often called completeness and soundness) for quantum money are well-known (see, e.g.,~\cite{AC12}.) Correctness asks that honestly generated money schemes should verify, i.e., $\Ver(\Mint)$ should always accept. Unforgeability states that an adversary with $k$ bills and oracle access to $\Ver$ should not be able to produce a state on which $\Ver^{\otimes k+1}$ accepts. In this work, we consider untraceable quantum money (also called ``quantum coins''~\cite{MS10}.) We give a formal security definition for untraceability, which states that an adversary $\algo A$ with oracle access to $\Ver$ and $\Mint$ cannot do better than random guessing in the following experiment: 
\begin{enumerate}
\item $\algo A$ outputs some candidate bill registers $\{M_j\}$ and a permutation $\pi$;
\item $b \from \{0,1\}$ is sampled, and if $b=1$ the registers $\{M_j\}$ are permuted by $\pi$; each candidate bill is verified and the failed ones are discarded;
\item $\algo A$ receives the rest of the bills and the entire internal state of the bank, and outputs a guess $b'$ for $b$.
\end{enumerate}
\begin{theorem}
The Haar money scheme \haarmoney, defined by setting
\begin{enumerate}
\item $\haarmoney.\Mint = \effstate(n, \negl(n)).\Gen$ 
\item $\haarmoney.\Ver = \effstate(n, \negl(n)).\Ver$
\end{enumerate}
is a correct quantum money scheme which satisfies information-theoretic unforgeability and untraceability.
\end{theorem}

One might reasonably ask if there are even stronger definitions of security for quantum money. Given its relationship to the ideal state sampler, we believe that Haar money should satisfy almost any notion of unforgeability and untraceability, including composable notions. We also remark that, based on the structure of the state simulator, which maintains an overall pure state supported on two copies of the symmetric subspace of banknote registers, it is straightforward to see that the scheme is also secure against an ``honest but curious'' or ``specious''~\cite{SSS09,DNS10} bank. We leave the formalization of these added security guarantees to future work.

\subsubsection{Sampling Haar-random unitaries.} 

Next, we turn to the problem of simulating Haar-random unitary operators. In this case, the ideal object $\idealunitary(n)$ initially samples a description of a perfectly Haar-random $n$-qubit unitary $U$, and then responds to two types of queries: $\idealunitary.\Eval$, which applies $U$, and $\idealunitary.\Invert$, which applies $U^\dagger$. In this case, we are able to construct a stateful simulator that runs in space polynomial in $n$ and the number of queries $q$, and is \emph{exactly indistinguishable} from $\idealunitary(n)$ to arbitrary adversaries. Our result can be viewed as a polynomial-space quantum analogue of the classical technique of lazy sampling for random oracles.

Our high-level approach is as follows. For now, suppose the adversary $\algo A$ only makes parallel queries to $\Eval$. If the query count $t$ of $\algo A$ is a priori bounded, we can simply sample an element of a unitary $t$-design. We can also do this coherently: prepare a quantum register $I$ in uniform superposition over the index set of the $t$-design, and then apply the $t$-design controlled on $I$. Call this efficient simulator $\effunitary_t$. Observe that the effect of $t$ parallel queries is just the application of the \emph{$t$-twirling channel} $\mathcal T^{(t)}$ to the $t$ input registers~\cite{BHH16}, and that $\effunitary_t$ simulates $\mathcal T^{(t)}$ faithfully. What is more, it applies a \emph{Stinespring dilation}\footnote{The Stinespring dilation of a quantum channel is an isometry with the property that the quantum channel can be implemented by applying the isometry and subsequently discarding an auxiliary register.}  \cite{Stinespring1955} of $\mathcal T^{(t)}$ with dilating register $I$. 

Now suppose $\algo A$ makes an ``extra'' query, i.e., query number $t+1$. Consider an alternative Stinespring dilation of $\mathcal T^{(t)}$, namely the one implemented by $\effunitary_{t+1}$ when queried $t$ times. 
Recall that all Stinespring dilations of a quantum channel are equivalent, up to a partial isometry on the dilating register. It follows that there is a partial isometry, acting on the private space of $\effunitary_t$, that transforms the dilation of $\mathcal T^{(t)}$ implemented by $\effunitary_t$ into the dilation of $\mathcal T^{(t)}$ implemented by $\effunitary_{t+1}$. If we implement this transformation, and then respond to $\algo A$ as prescribed by $\effunitary_{t+1}$, we have achieved perfect indistinguishability against the additional query. By iterating this process, we see that the a priori bound on the number of queries is no longer needed. We let $\effunitary$ denote the resulting simulator. The complete construction is described in \expref{Construction}{con:unitary-sampler} below.

Our high-level discussion above did not take approximation into account. All currently known efficient constructions of $t$-designs are approximate. Here, we take a different approach: we will implement our construction using \emph{exact $t$-designs}. This addresses the issue of adaptive queries: if there exists an adaptive-query distinguisher with nonzero distinguishing probability, then by post-selection there also exists a parallel-query one via probabilistic teleportation. This yields that the ideal and efficient unitary samplers are perfectly indistinguishable to arbitrary adversaries.

\begin{theorem}\label{thm:unitary-sampler-summary}
For all oracle algorithms $\algo A$,
$	
\Pr\left[\algo A^{\idealunitary(n)} = 1\right]=\Pr\left[\algo A^{\effunitary(n)} = 1\right].
$
\end{theorem}

The existence of exact unitary $t$-designs for all $t$ is a fairly recent result. It follows as a special case of a result of Kane \cite{Kane2015}, who shows that designs exist for all finite-dimensional vector spaces of well-behaved functions on path-connected topological spaces. He also gives a simpler result for homogeneous spaces when the vector space of functions is invariant under the symmetry group action. Here, the number of elements of the smallest design is bounded just in terms of the dimension of the space of functions. The unitary group is an example of such a space, and the dimension of the space of homogeneous polynomials of degree $t$ in both $U$ and $U^\dagger$ can be explicitly derived, see e.g. \cite{Roy2009}. This yields the following.
\begin{corollary}
	The space complexity of $\effunitary(n)$ for $q$ queries is bounded from above by  $2q(2n+\log e)+O(\log q)$.
\end{corollary}

\subsubsection{An alternative approach.}

We now sketch another potential approach to lazy sampling of unitaries. Very briefly, this approach takes a representation-theoretic perspective and suggests that the Schur transform~\cite{Bacon2006} could lead to a polynomial-time algorithm for lazy sampling Haar-random unitaries. The discussion below uses tools and language from quantum information theory and the representation theory of the unitary and symmetric groups to a much larger extent than the rest of the article, and is not required for understanding our main results.

We remark that the analogous problem of lazy sampling a quantum oracle for a random classical function was recently solved by Zhandry~\cite{Zhandry2019}. One of the advantages of Zhandry's technique is that it partly recovers the ability to inspect previously made queries, an important feature of classical lazy sampling. The key insight is that the simulator can implement the Stinespring dilation of the oracle channel, and thus record the output of the \emph{complementary channel}.\footnote{The complementary channel of a quantum channel maps the input to the auxiliary output of the Stinespring dilation isometry.}  As the classical function is computed via XOR, changing to the $\Z_2^n$-Fourier basis makes the recording property explicit. It also allows for an efficient implementation.

In the case of Haar-random unitary oracles, we can make an analogous observation. Consider an algorithm that makes $t$ parallel queries to $U$. The relevant Fourier transform is now over the unitary group, and is given by the Schur transform~\cite{Bacon2006}. By Schur-Weyl duality (see e.g.~\cite{Christandl2006}), the decomposition of $\left(\mathbb C^{2^n}\right)^{\otimes t}$ into irreducible representations is given by
\begin{equation}\label{eq:SchW}
\left(\C^d\right)^{\otimes t}\cong \bigoplus_{\lambda\vdash_d t}[\lambda]\otimes V_{\lambda,d}.
\end{equation}
Here $\lambda\vdash_d t$ means $\lambda$ is any partition of $t$ into at most $d$ parts, $[\lambda]$ is the Specht module of $S_t$, and $V_{\lambda,d}$ is the Weyl module of $U(d)$, corresponding to the partition $\lambda$, respectively. By Schur's lemma, the $t$-twirling channel acts as
\begin{equation}\label{eq:twirly}
	\mathcal T^{(t)}=\bigoplus_{\lambda\vdash_d t}\mathrm{id}_{[\lambda]}\otimes \Lambda_{V_{\lambda,d}},
\end{equation}
where $\mathrm{id}$ is the identity channel, and $\Lambda=\Tr(\cdot)\tau$ with the maximally mixed state $\tau$ is the depolarizing channel. We therefore obtain a Stinespring dilation of the $t$-twirling channel as follows. Let $\tilde B, \tilde B'$ be registers with Hilbert spaces 
\begin{equation}
	\Hi_{\tilde B}=	\Hi_{\tilde B'}=\bigotimes_{\lambda\vdash_d t}V_{\lambda,d}
\end{equation} 
and denote the subregisters by $\tilde B_{\lambda}$ and $\tilde B'_\lambda$, respectively. Let further $\ket{\phi^+}_{\tilde B\tilde B'}$ be the standard maximally entangled state on these registers, and let $C$ be a register whose dimension is the number of partitions of $t$ (into at most $2^n$ parts). Define the isometry
\begin{equation}
 	\hat V_{A^t \tilde B\to  A^t \tilde B C}=\bigoplus_{\lambda\vdash_d t} F_{V_{\lambda,d} \tilde B_\lambda}\otimes \id_{[\lambda]}\otimes \ket{\lambda}_C
\end{equation}
In the above equation $V_{\lambda,d}$ and $[\lambda]$ are understood to be subspaces of $A^t$, the identity operators on $\tilde B_\mu$, $\mu\neq\lambda$ are omitted and $F$ is the swap operator. By \eqref{eq:twirly}, a Stinespring dilation of the $t$-twirling channel is then given by
\begin{equation}
	V_{A^t\to A^t\tilde B\tilde B' C}=\hat V_{A^t \tilde B\to  A^t \tilde B C}\ket{\phi^+}_{\tilde B\tilde B'}.
\end{equation}
By the equivalence of all Stinespring dilations, the exists an isometry $W_{\hat B_t\to \tilde B\tilde B' C}$ that transforms the state register of $\effunitary(n)$ after $t$ parallel queries so that the global state is the same as if the Stinespring dilation above had been applied to the $t$ input registers. But now the quantum information that was contained in the subspace $V_{\lambda, d}$ of the algorithm's query registers can be found in register $\tilde B_\lambda$.

\subsection{Organization}

The remainder of the paper is organized as follows. In \expref{Section}{sec:prelims}, we recall some basic notation and facts, and some lemmas concerning coherent preparation of certain generic families of quantum states. The proofs for these lemmas are given in \expref{Appendix}{sec:prep-lemmas}. We also describe stateful machines, which will be our model for thinking about the aforementioned ideal objects and their efficient simulators. In \expref{Section}{sec:state} we describe our efficient simulator for Haar-random states, and in \expref{Section}{sec:unitary} we describe our polynomial-space simulator for Haar-random unitaries. We end by describing the Haar money scheme and establishing its security in \expref{Section}{sec:money}.

\subsection{Acknowledgments}
%
The authors thank Yi-Kai Liu, Carl Miller, and Fang Song on helpful comments on an earlier draft. CM thanks Michael Walter for discussions about $t$-designs. CM was funded by a NWO VIDI grant (Project No. 639.022.519) and a NWO VENI grant (Project No. VI.Veni.192.159). GA acknowledges support from NSF grant CCF-1763736.

\section{Preliminaries}\label{sec:prelims}

\subsection{Some basics}

Given a fixed-size (e.g., $n$-qubit) register $A$, we will use $A_1, A_2, \dots$ to denote indexed copies of $A$. We will use $A^t$ to denote a register consisting of $t$ indexed copies of $A$, i.e., $A^t = A_1 A_2 \cdots A_t$.  Unless stated otherwise, distances of quantum states are measured in the trace distance, i.e.,
$$
d(\rho, \sigma)=\frac 1 2\| \rho-\sigma\|_1
\qquad \text{where} \qquad
\| X\|_1=\tr{\sqrt{X^\dagger X}}\,.
$$
Distances of unitary operators are measured in the operator norm.

We will frequently apply operators to some subset of a larger collection of registers. In that context, we will use register indexing to indicate which registers are being acted upon, and suppress identities to simplify notation. The register indexing will also be suppressed when it is clear from context. For example, given an operator $X_{A\to B}$ and some state $\rho$ on registers $A$ and $C$, we will write $X(\rho)$ in place of $(X \otimes \one_C) (\rho)$ to denote the state on $BC$ resulting from applying $X$ to the $A$ register of $\rho$.

We let $\ket{\phi^+}_{AA'}$ denote the maximally entangled state on registers $A$ and $A'$. For a linear operator $X$ and some basis choice, we denote its transpose by $X^T$.

\begin{lemma}[Mirror lemma; see, e.g.,~\cite{M18}]\label{lem:mirror} 
For $X_{A \rightarrow B}$ a linear operator,
$$
X_{A \rightarrow B} \ket{\phi^+}_{AA'} = \sqrt{\frac{\dim(B)}{\dim(A)}} X^T_{B' \rightarrow A'} \ket{\phi^+}_{BB'}\,.
$$
\end{lemma}

\subsection{Unitary designs}

Let $\mu_n$ be the Haar measure on the unitary group $\mathrm{U}(2^n)$. We define the Haar $t$-twirling channel $\mathcal T^{(t)}_{\mathrm{Haar}}$ by
\begin{equation}\label{eq:twirl}
	\mathcal T^{(t)}_{\mathrm{Haar}}(X)=\int_{\mathrm{U}(2^n)} U^{\otimes t}X\left(U^{\otimes t}\right)^\dagger\mathrm d\mu(U).
\end{equation}
For a finite subset $D\subset \mathrm{U}(2^n)$, we define the $t$-twirling map with respect to $D$ as
\begin{equation}\label{eq:settwirl}
\mathcal T^{(t)}_{D}(X)=\frac{1}{|D|}\sum_{U\in D} U^{\otimes t}X\left(U^{\otimes t}\right)^\dagger.
\end{equation}
An $n$-qubit unitary $t$-design is a finite set $D\subset \mathrm{U}(2^n)$ such that
\begin{equation}
	\mathcal T^{(t)}_{D}=\mathcal T^{(t)}_{\mathrm{Haar}}(X)
\end{equation}
%

Another twirling channel is the mixed twirling channels with $\ell$ applications of the unitary and $t-\ell$ applications of it's inverse, 
\begin{equation}\label{eq:mixed}
\mathcal T^{(\ell, t-\ell)}_{\mathrm{Haar}}(\Gamma)=\int_{\mathrm{U}(2^n)} U^{\otimes \ell}\otimes\left(U^{\otimes(t-\ell)}\right)^\dagger \Gamma\left(U^{\otimes \ell}\right)^\dagger\otimes U^{\otimes(t-\ell)}\mathrm d\mu(U).
\end{equation}
The mixed twirling channel $\mathcal T^{(\ell, t-\ell)}_D$ for a finite set $D\subset\mathrm{U}(2^n)$ is also defined analogous to Equation \eqref{eq:settwirl}. As our definition of unitary $t$-designs is equivalent to one based on the expectation values of polynomials (see, e.g.,~\cite{L10}), we easily obtain the following.
\begin{proposition}\label{lem:mixedtwirl}
		Let $D$ be an $n$-qubit unitary $t$-design and $0\le\ell\le t$. Then
		\begin{equation}
		\mathcal T^{(\ell, t-\ell)}_{\mathrm{Haar}}=\mathcal T^{(\ell,t-\ell)}_{D}
		\end{equation}
\end{proposition}

 Finite exact unitary $t$-designs exist. In particular, one can apply the following theorem to obtain an upper bound on their minimal size. Here, a design for a function space $W$ on a topological space $X$ with measure $\mu$ is a finite set $D \subset X$ such that the expectation of a function $f\in W$ is the same whether it is taken over $X$ according to $\mu$ or over the uniform distribution on $D$.
\begin{theorem}[\cite{Kane2015}, Theorem 10]
	Let $X$ be a homogeneous space, $\mu$ an invariant measure 
on $X$ and $W$ a $M$-dimensional vector subspace of the space of real functions on $X$ that is invariant under the symmetry group of $X$, where $M>1$. Then for any $N>M(M-1)$, there exists a $W$-design for $X$ of size $N$.
	Furthermore, there exists a design for $X$ of size at most $M(M-1)$.
\end{theorem}
The case of unitary $t$-designs is the one where $X=\mathrm{U}(2^n)$ is acting on itself (e.g., on the left), $\mu$ is the Haar measure, and $W$ is the vector space of homogeneous polynomials of degree $t$ in both $U$ and $U^\dagger$\footnote{The output of the twirling channel \eqref{eq:twirl} is a matrix of such polynomials.}. The dimension of this space is
\begin{equation}
M_t=\left(\begin{array}{c}
2^{2n}+t-1\\t
\end{array}\right)^2\le\left( \frac{e(2^{2n}+t-1)}{t}\right)^t,
\end{equation}
see e.g. \cite{Roy2009}. We therefore get
\begin{corollary}\label{cor:udesigns}
For all $n$, there exists an exact $n$-qubit unitary $t$-design with a number of elements which is at most
$$
\left( \frac{e(2^{2n}+t-1)}{t}\right)^{2t}\,.
$$
\end{corollary}

\subsection{Real and ideal stateful machines}

We will frequently use stateful algorithms with multiple ``interfaces'' which allow a user to interact with the algorithm. We will refer to such objects as \emph{stateful machines}. We will use stateful machines to describe functionalities (and implementations) of collections of oracles which relate to each other in some way. For example, one oracle might output a fixed state, while another oracle reflects about that state. 
\begin{definition}[Stateful machine]
A stateful machine $\algo S$ consists of:
\begin{itemize}
\item A finite set $\Lambda$, whose elements are called interfaces. Each interface $\mathcal I \in \Lambda$ has two fixed parameters $n_\mathcal I \in \mathbb N$ (input size) and $m_\mathcal I \in \mathbb N$ (output size), and a variable $t_\algo I$ initialized to $1$ (query counter.)
\item For each interface $\mathcal I \in \Lambda$, a sequence of quantum algorithms $\{\algo S.\mathcal I_j : j = 1, 2, \dots\}$. Each $\algo S.\mathcal I_j$ has an input register of $n_\mathcal I$ qubits, an output register of $m_\mathcal I$ qubits, and is allowed to act on an additional shared work register $R$ (including the ability to add/remove qubits in $R$.) In addition, each $\algo S.\algo I_j$ increments the corresponding query counter $t_\algo I$ by one.
\end{itemize}
\end{definition}
The typical usage of a stateful machine $\algo S$ is as follows. First, the work register $R$ is initialized to be empty, i.e., no qubits. After that, whenever a user invokes an interface $\algo S.\algo I$ and supplies $n_\algo I$ qubits in an input register $M$, the algorithm $\algo S.\algo I_{t_\algo I}$ is invoked on registers $M$ and $R$. The contents of the output register are returned to the user, and the new, updated work register remains for the next invocation. We emphasize that the work register is shared between all interfaces.

We remark that we will also sometimes define \emph{ideal machines}, which behave outwardly like a stateful machine but are not constrained to apply only maps which are implementable in finite space or time. For example, an ideal machine can have an interface that implements a perfectly Haar-random unitary $U$, and another interface which implements $U^\dagger$.

\subsection{Some state preparation tools}

We now describe some algorithms for efficient coherent preparation of certain quantum state families. The proofs for the following lemmas can be found in \expref{Appendix}{sec:prep-lemmas}. We begin with state families with polynomial support.

\begin{lemma}\label{lem:poly-prepare}
Let $\ket{\varphi} = \sum_{x \in \{0,1\}^n} \varphi(x) \ket{x}$ be a family of quantum states whose amplitudes $\varphi$ have an efficient classical description $\tilde \varphi$, and such that $|\{x : \varphi(x) \neq 0\}| \leq \poly$. Then there exists a quantum algorithm $\mathcal P$ which runs in time polynomial in $n$ and $\log(1/\epsilon)$ and satisfies
$$
\|\mathcal P \ket{\tilde \varphi}\ket{0^n} - \ket{\tilde \varphi}\ket{\varphi}\|_2 \leq \epsilon\,.
$$
\end{lemma}

Given a set $S \subset \bits^n$, we let
$$
\ket{S} := \frac{1}{\sqrt{|S|}}\sum_{x \in S} \ket x
\qquad \text{and} \qquad
\ket{\bar S} := \frac{1}{\sqrt{2^n-|S|}}\sum_{x \in \{0,1\}\setminus S} \ket{x} 
$$
denote the states supported only on $S$ and its set complement $\bar S$, respectively. Provided that $S$ has polynomial size, we can perform coherent preparation of both state families efficiently: the former by \expref{Lemma}{lem:poly-prepare} and the latter via the below.

\begin{lemma}\label{lem:almost-uniform-prepare}
Let $S \subset \{0,1\}^n$ be a family of sets of size $\poly$ with efficient description $\tilde S$, and let $\epsilon > 0$. There exists a quantum algorithm $\mathcal P$ which runs in time polynomial in $n$ and $\log(1/\epsilon)$ and satisfies
$$
\left\|\mathcal P \ket{\tilde S}_A\ket{0^n}_B - \ket{\tilde S}_A\ket{\bar S}_B\right\|_2 \leq \epsilon\,.
$$
\end{lemma}

Finally, we show that if two orthogonal quantum states can be prepared, then so can an arbitrary superposition of the two.

\begin{lemma}\label{lem:super-prepare}
Let $\ket{\zeta_{0,j}}, \ket{\zeta_{1,j}}$ be two familes of $n$-qubit quantum states such that $\braket{\zeta_{0,j}}{\zeta_{1,j}}=0$ for all $j$, and such that there exists a quantum algorithm $\algo P_b$ which runs in time polynomial in $n$ and $\log(1/\epsilon)$ and satisfies 
$
\|\mathcal P_b \ket{j}\ket{0^n} - \ket{j}\ket{\zeta_{b, j}}\|_2 \leq \epsilon
$
for $b \in \{0,1\}$. 

For $z_0, z_1\in\C$ such that $|z_0|^2+|z_1|^2=1$, let $\tilde z$ denote a classical description of $(z_0,z_1)$ to precision at least $\epsilon$. There exists a quantum algorithm $\algo Q$ which runs in time polynomial in $n$ and $\log(1/\epsilon)$ and satisfies
	\begin{equation}
		\left\|\algo Q \ket{j}\ket{\tilde z}\ket{0^n} - \ket{j}\ket{\tilde z}\bigl(z_0\ket{\zeta_{0,j}} + z_1\ket{\zeta_{1,j}}\bigr)\right\|_2 \leq \epsilon\,.
	\end{equation}
\end{lemma}

\section{Simulating a Haar-random state oracle}\label{sec:state}

\subsection{The problem, and our approach}



We begin by defining the ideal object we'd like to emulate. Here we deviate slightly from the discussion above, in that we ask for the reflection oracle to also accept a (quantum) control bit.


\begin{construction}[Ideal state sampler]\label{con:state-ideal}
The ideal $n$-qubit state sampler is an ideal machine $\idealstate(n)$ with interfaces $(\Init, \Gen, \Ver, \CReflect)$, defined as follows. 
\begin{enumerate}
\item $\idealstate(n).\Init:$ takes no input; samples a description $\tilde \varphi$ of an $n$-qubit state $\ket{\varphi}$ from the Haar measure.
\item $\idealstate(n).\Gen:$ takes no input; uses $\tilde \varphi$ to prepare a copy of $\ket{\varphi}$ and outputs it.
\item $\idealstate(n).\Ver:$ receives $n$-qubit input; uses $\tilde \varphi$ to apply the measurement $\{\proj{\varphi},$ $ \one - \proj{\varphi}\}$; return the post-measurement state and output $\acc$ in the first case and $\rej$ in the second. 
\item $\idealstate(n).\CReflect:$ receives $(n+1)$-qubit input; uses $\tilde \varphi$ to implement the controlled reflection $R_\varphi := \ketbra{0}{0} \otimes \one + \ketbra{1}{1} \otimes (\one - 2\ketbra{\varphi}{\varphi})$ about $\ket{\varphi}$.
\end{enumerate}
\end{construction}

We assume that $\Init$ is called first, and only once; the remaining oracles can then be called indefinitely many times, and in any order. If this is inconvenient for some application, one can easily adjust the remaining interfaces to invoke $\Init$ if that has not been done yet. We remark that $\Ver$ can be implemented with a single query to $\CReflect$. 

\begin{lemma}\label{lem:ver-reflect}
	$\Ver$ can be simulated with one application of $\CReflect$.
\end{lemma} 
\begin{proof}
Prepare an ancillary qubit in the state $\ket +$ and apply the reflection on the input controlled on the ancillary qubit. Then apply $H$ to the ancilla qubit and measure it. Output all the qubits, with the ancilla interpreted as $1 = \acc$ and $0 = \rej$.
\qed \end{proof}

Our goal is to devise a stateful simulator for \expref{Construction}{con:state-ideal} which is efficient. Efficient here means that, after $t$ total queries to all interfaces (i.e., $\Init$, $\Gen$, $\Ver$, and $\CReflect$), the simulator has expended time polynomial in $n$, $t$, and $\log (1/\epsilon)$. 


As described in \expref{Section}{sec:intro-state}, our approach will be to ensure that, for every $t$, the state shared between the adversary $\algo A$ and our stateful oracle simulator $\effstate$ will be maximally entangled between two copies of the $t$-fold symmetric subspace $\BSym_{n, t}$: one held by $\algo A$, and the other by $\effstate$. The extension from the $t$-fold to the $(t+1)$-fold joint state will be performed by an isometry $V^{t \rightarrow t+1}$ which acts only on the state of $\effstate$ and two fresh $n$-qubit registers $A_{t+1}$ and $B_{t+1}$ initialized by $\effstate$. After $V$ is applied, $A_{t+1}$ will be given to $\algo A$. As we will show, $V$ can be performed efficiently using some algorithmic tools for working with symmetric subspaces, which we will develop in the next section. This will yield an efficient way of simulating $\Gen$. Simulation of $\Ver$ and $\CReflect$ will follow without much difficulty, as outlined in \expref{Section}{sec:intro-state}.

\subsection{Some tools for symmetric subspaces}

\subsubsection{A basis for the symmetric subspace.}

We recall an explicit orthonormal basis of the symmetric subspace (see, e.g., \cite{Ji2018} or \cite{Harrow13}.)  Let 
\begin{equation}
	\snt_{n,t}=\left\{\alpha\in\left(\bits^n\right)^t\Big|\alpha_1\le \alpha_2\le...\le\alpha_t\right\}
\end{equation}
be the set of lexicographically-ordered $t$-tuples of $n$ bit strings. For each $\alpha\in \snt_{n,t}$, define the unit vector
\begin{equation}
	\ket{\Sym(\alpha)}=\left(t!\prod_{x\in\bits^n}f_x(\alpha)!\right)^{-\frac 1 2}\sum_{\sigma\in S_t}\ket{\alpha_{\sigma(1)}}\ket{\alpha_{\sigma(2)}}...\ket{\alpha_{\sigma(t)}}.
\end{equation}
Here, $f_x(\alpha)$ is the number of times the string $x$ appears in the tuple $\alpha$. The set $\{\ket{\Sym(\alpha)} : \alpha \in \snt_{n,t}\}$ is an orthonormal basis for $\Sym^{t}\C^{2^n}$. We remark that the Schmidt decomposition of $\ket{\Sym(\alpha)}$ with respect to the bipartition formed by the $t$-th register vs. the rest is given by
\begin{equation}\label{eq:t-schmidt}
	\ket{\Sym(\alpha)}=\sum_{x\in\bits^n}\sqrt{\frac{f_x(\alpha)}{t}}\ket{\Sym(\alpha^{-x})}\ket x,
\end{equation}
where $\alpha^{-x}\in \snt_{n,t-1}$ is the tuple $\alpha$ with one copy of $x$ removed.

\subsubsection{Some useful algorithms.}

We now describe some algorithms for working in the above basis. Let $A$ and $B$ denote $n$-qubit registers. Recall that $A_j$ denotes indexed copies of $A$ and that $A^t$ denotes $A_1 A_2 \cdots A_t$, and likewise for $B$. In our setting, the various copies of $A$ will be prepared by the oracle simulator and then handed to the query algorithm at query time. The copies of $B$ will be prepared by, and always remain with, the oracle simulator.

\begin{proposition}
For each $n$, $t$ and $\epsilon=2^{-\mathrm{poly}(n, t)}$, there exists an efficiently implementable unitary $U^\Sym_{n, t}$ on $A^t$ such that for all $\alpha \in \snt_{n,t}$, $U^\Sym_{n, t} \ket{\alpha} = \ket{\Sym(\alpha)}$ up to trace distance $\epsilon$.
\end{proposition}
\begin{proof} 
Clearly, the operation 
\begin{equation}\label{eq:sort}
\ket{\Sym(\alpha)}\ket \beta\mapsto\ket{\Sym(\alpha)}\ket {\beta\oplus \alpha}
\end{equation}
is efficiently implementable exactly, by XORing the classical sort function of the first register into the second register. 

Let us now show that the operation $\ket{\alpha}\mapsto \ket\alpha\ket{\Sym(\alpha)}$ is also efficiently implementable (up to the desirable error) by exhibiting an explicit algorithm. We define it recursively in $t$, as follows. For $t=1$, $\Sym(x)=x$ for all $x\in\bits^n$, so this case is simply the map $\ket{x} \mapsto \ket{x}\ket{x}$. Suppose now the operation $\ket{\alpha}\mapsto\ket\alpha\ket{\Sym(\alpha)}$ can be implemented for any $\alpha\in \snt_{n,t-1}$. The $t$-th level algorithm will begin by applying
$$
\ket\alpha\mapsto \ket\alpha\sum_{x\in\bits^n}\sqrt{\frac{f_x(\alpha)}{t}}\ket x\,.
$$ 
Since $f_x(\alpha)$ is nonzero for only $t$-many $x \in \bits^n$, this can be implemented efficiently by \expref{Lemma}{lem:poly-prepare}. Next, we perform $\ket\alpha \ket x\mapsto\ket \alpha\ket x\ket{\alpha^{-x}}$. Using the algorithm for $t-1$, we then apply $\ket \alpha\ket x\ket{\alpha^{-x}}\mapsto \ket \alpha\ket x\ket{\alpha^{-x}}\ket{\Sym(\alpha^{-x})}$, and uncompute $\alpha^{-x}$. By \eqref{eq:t-schmidt}, we have in total applied $\ket{\alpha}\mapsto\ket\alpha \ket{\Sym(\alpha)}$ so far. To finish the $t$-th level algorithm for approximating $\ket{\alpha} \mapsto \ket{\Sym(\alpha)}$, we simply apply \eqref{eq:sort} to uncompute $\alpha$ from the first register.
\qed \end{proof}

\begin{theorem}[Restatement of \expref{Theorem}{thm:sym-intro}]\label{thm:sym-increment}
For each $n$, $t$ and $\epsilon=2^{-\mathrm{poly}(n, t)}$, there exists an efficiently implementable isometry $V^{t \to t+1}$ from $B^t$ to $A_{t+1}B^{t+1}$ such that, up to trace distance $\epsilon$,
$$
V : \sum_{\alpha \in {S_{n,t}^\uparrow}} \ket{\Sym(\alpha)}_{A^t}\ket{\Sym(\alpha)}_{B^t}
\longmapsto \sum_{\beta \in {S_{n,t+1}^\uparrow}} \ket{\Sym(\beta)}_{A^{t+1}}\ket{\Sym(\beta)}_{B^{t+1}}\,.
$$
\end{theorem}
\begin{proof} 
We describe the algorithm assuming all steps can be implemented perfectly. It is straightforward to check that each step we use can in reality be performed to a sufficient accuracy that the accuracy of the entire algorithm is at least $\epsilon$. 

We will need a couple of simple subroutines. First, given $\alpha \in \snt_{n,t}$ and $x \in \{0,1\}^n$, we define $\alpha^{+x}$ to be the element of $\snt_{n, t+1}$ produced by inserting $x$ at the first position such that the result is still lexicographically ordered. One can perform this reversibly via $\ket{\alpha}\ket{0^n}\ket{x} \mapsto \ket{\alpha}\ket{x}\ket{x} \mapsto \ket{\alpha^{+x}}\ket{x}$.

Second, we will need to do coherent preparation of the state
\begin{equation}
\ket{\psi_\alpha}=\sum_{x\in\bits^n}\sqrt{\frac{1+f_x(\alpha)}{2^n+t}}\ket x\,.
\end{equation}
For any given $\alpha \in \snt_{n, t}$, the state $\ket{\psi_\alpha}$ can be prepared by using the preparation circuit for the two orthogonal components of the state whose supports are $\{x : f_x(\alpha) > 0\}$ and $\{x : f_x(\alpha) = 0\}$. These two components can also be prepared coherently using \expref{Lemma}{lem:poly-prepare} and \expref{Lemma}{lem:almost-uniform-prepare}, respectively. Their superposition can be prepared with \expref{Lemma}{lem:super-prepare}. Putting it all together, we get an algorithm for $\ket{\alpha}\ket{0^n} \mapsto \ket{\alpha}\ket{\psi_\alpha}$.

The complete algorithm is a composition of several efficient routines. We describe this below, explicitly calculating the result for the input states of interest. For readability, we omit overall normalization factors.
\begin{align*}
    &\sum_{\alpha} \ket{\Sym(\alpha)}_{A^t}\ket{\Sym(\alpha)}_{B^t} &&~\\
\longmapsto
&\sum_{\alpha} \ket{\Sym(\alpha)}_{A^t} \ket{0^n} \ket{\Sym(\alpha)}_{B^t} \ket{0^n}
 && \text{add working registers}\\
\longmapsto
&\sum_{\alpha} \ket{\Sym(\alpha)}_{A^t}  \ket{0^n} \ket{\alpha}_{B^t} \ket{0^n}
 && \text{apply $\bigl(U^\Sym_{n, t}\bigr)^\dagger$ to $B^t$}\\
 \longmapsto
 &\sum_{\alpha, x} \sqrt{\frac{1+f_x(\alpha)}{2^n+t}}\ket{\Sym(\alpha)}_{A^t} \ket{x}  \ket{\alpha}_{B^t} \ket{0^n}
 && \text{prepare }\ket{\psi_{\alpha}}\\
\longmapsto
&\sum_{\alpha, x} \sqrt{\frac{1+f_x(\alpha)}{2^n+t}}\ket{\Sym(\alpha)}_{A^t} \ket{x} \ket{\alpha^{+x}}_{B^{t+1}} 
 && \text{insert $x$ into $\alpha$}\\
\longmapsto
&\sum_{\alpha, x} \sqrt{\frac{1+f_x(\alpha)}{2^n+t}}\ket{\Sym(\alpha)}_{A^t}\ket{x}_{A_{t+1}}\ket{\Sym(\alpha^{+x})}_{B^{t+1}} 
 && \text{apply $U^\Sym_{n, t+1}$ to $B^{t+1}$}
\end{align*}
To see that the last line above is the desired result, we observe that we can index the sum in the last line above in a more symmetric fashion: the sum is just taken over all pairs $(\alpha, \beta)$ such that the latter can be obtained from the former by adding one entry (i.e., the string $x$). But that is the same as summing over all pairs $(\alpha, \beta)$, such that the former can be obtained from the latter by \emph{removing} one entry.
\begin{align*}
&\quad\,\sum_{\alpha, x} \sqrt{\frac{1+f_x(\alpha)}{2^n+t}}\ket{\Sym(\alpha)}_{A^t}\ket{x}_{A_{t+1}}\ket{\Sym(\alpha^{+x})}_{B^{t+1}}\\
 &=\sum_{\beta, x} \sqrt{\frac{f_x(\beta)}{2^n+t}}\ket{\Sym(\beta^{-x})}_{A^t}\ket{x}_{A_{t+1}}\ket{\Sym(\beta)}_{B^{t+1}}\\
&=\sqrt{\frac{t}{2^n+t}}\sum_{\beta} \left(\sum_x\sqrt{\frac{f_x(\beta)}{t}}\ket{\Sym(\beta^{-x})}_{A^t}\ket{x}_{A_{t+1}}\right)\ket{\Sym(\beta)}_{B^{t+1}}\\
&=\sqrt{\frac{t}{2^n+t}}\sum_{\beta} \ket{\Sym(\beta)}_{A^{t+1}}\ket{\Sym(\beta)}_{B^{t+1}}.
\end{align*}
Here, the last equality is \eqref{eq:t-schmidt}, and the prefactor is the square root of the quotient of the dimensions of the $t$- and $(t+1)$-copy symmetric subspaces, as required for a correct normalization of the final maximally entangled state.
\qed \end{proof}

\subsection{State sampler construction and proof}


\begin{construction}[Efficient state sampler]\label{con:state-sampler}
Let $n$ be a positive integer and $\epsilon$ a negligible function of $n$. The efficient $n$-qubit state sampler with precision $\epsilon$ is a stateful machine $\effstate(\epsilon,n)$ with interfaces $(\Init, \Gen, \Reflect)$, defined below. For convenience, we denote the query counters by $t = t_\Gen$ and $q = t_\Reflect$ in the following.
\begin{enumerate}
\item $\effstate(\epsilon,n).\Init:$ prepares the standard maximally entangled state $\ket{\phi^+}_{A_1B_1}$ on $n$-qubit registers $A_1$ and $B_1$, and stores both $A_1$ and $B_1$.
\item $\effstate(\epsilon,n).\Gen:$ On the first query, outputs register $A_1$. On query $t$, takes as input registers $B^{t-1}$ and produces registers $A_tB^t$ by applying the isometry $V^{t-1 \to t}$ from \expref{Theorem}{thm:sym-increment} with accuracy $\epsilon2^{-(t+2q)}$; then it outputs $A_t$ and stores $B^t$.
\item $\effstate(\epsilon,n).\CReflect:$ On query $q$ with input registers $CA^*$, do the following controlled on the qubit register $C$: apply $\left(U^{t-1 \to t}\right)^\dagger$, a unitary implementation of $V^{t-1 \to t}$, with accuracy $\epsilon2^{-(t+2(q-1))}$, in the sense that $V^{t-1 \to t}=U^{t-1 \to t}\ket{0^{2n}}_{A_tB_t}$, with $A^*$ playing the role of $A_t$. Subsequently, apply a phase $-1$ on the all-zero state of the ancilla registers $A_t$ and $B_t$, and reapply $U^{t-1 \to t}$, this time with accuracy $\epsilon2^{-(t+2(q-1)+1)}$. 
\end{enumerate}
\end{construction}

We omitted defining $\effstate.\Ver$ since it is trivial to build from $\CReflect$, as described in \expref{Lemma}{lem:ver-reflect}. By \expref{Theorem}{thm:sym-increment}, the runtime of $\effstate(\epsilon,n)$ is polynomial in $n$, $\log(1/\epsilon)$ and the total number of queries $q$ that are made to its various interfaces.

We want to show that the above sampler is indistinguishable from the ideal sampler to any oracle algorithm, in the following sense. Given a stateful machine $\algo C \in \{\idealstate(n), \effstate(n, \epsilon)\}$ and a (not necessarily efficient) oracle algorithm $\algo A$, we define the process $b \from \algo A^{\algo C}$ as follows:
\begin{enumerate}
\item $\algo C.\Init$ is called;
\item $\algo A$ receives oracle access to $\algo C.\Gen$ and $\algo C.\CReflect$;
\item $\algo A$ outputs a bit $b$\,.
\end{enumerate}
\begin{theorem}\label{thm:state-sampler}
For all oracle algorithms $\algo A$ and all $\epsilon>0$ that can depend on $n$ in an arbitrary way,
\begin{equation}\label{eq:state-thm}
\left|\Pr\left[\algo A^{\idealstate(n)} = 1\right] - \Pr\left[\algo A^{\effstate(n, \epsilon)} = 1\right]\right| \leq \epsilon\,. 
\end{equation}
\end{theorem}
\begin{proof}

During the execution of $\effstate(\epsilon, n)$, the $i$-th call of $V^{t-1 \to t}$ (for any $t$) incurs a trace distance error of at most $\epsilon2^{-i}$. The trace distance between the outputs of $\algo A^\effstate(\epsilon, n)$ and $\algo A^\effstate(0, n)$ is therefore bounded by 
$\sum_{i=1}^\infty \epsilon2^{-i}=\epsilon$. It is thus sufficient to establish the theorem for $\effstate(0, n)$. 

For any fixed $q$, there exists a stateful machine $\hat\effstate(0, q, n)$ which is perfectly indistinguishable from $\idealstate(n)$ to all adversaries who make a maximum total number $q$ of queries. The \Init procedure of $\hat\effstate(0, q, n)$ samples a random element $U_i$ from an exact unitary $2q$-design $D^{2q}=\{U_i\}_{i\in I}$. Queries to \Gen are answered with a copy of $U_i\ket 0$, and \Reflect is implemented by applying $\mathds 1-2U_i\proj 0U_i^\dagger$. It will be helpful to express $\hat\effstate(0, q, n)$ in an equivalent isometric form. In this form, the initial oracle state is
	\begin{equation}
	\ket\eta=|I|^{-1/2}\sum_{i\in I}\ket i_{\hat B}\,.
	\end{equation}
\Gen queries are answered using the $\hat B$-controlled isometry 
	\begin{equation}
	\hat V^{t\to t+1}_{\hat B\to \hat B A_{t+1}}=\sum_{i\in I}\proj i_{\hat B}\otimes U_i\ket 0_{A_{t+1}}\,.
	\end{equation}
\Reflect queries are answered by
	\begin{align}
	\hat V^{\Reflect}_{\hat BA^*\to \hat B A^*}=&\mathds 1-2\sum_{i\in I}\proj i_{\hat B}\otimes U_i\proj 0_{A^*} U_i^\dagger\\
	=&\mathds 1-2  \hat V^{t\to t+1}_{\hat B\to \hat B A^*}\left(\hat V^{t\to t+1}\right)^\dagger_{\hat BA^*\to \hat B }\,.
	\end{align}

Now suppose $\algo A$ is an arbitrary (i.e., not bounded-query) algorithm making only $\Gen$ queries. We will show that after $q$ queries, the oracles $\effstate(0,n)$ and $\hat\effstate(0, q, n)$ are equivalent, and that this holds for all $q$. We emphasize that $\effstate(0,n)$ does not depend on $q$; as a result, we can apply the equivalence for the appropriate total query count $q_\mathsf{total}$ after $\algo A$ has produced its final state, even if $q_\mathsf{total}$ is determined only at runtime. It will thus follow that $\effstate(0,n)$ is equivalent to $\idealstate(n)$.

To show the equivalence betwen $\effstate(0,n)$ and $\hat\effstate(0,q,n)$, we will demonstrate a partial isometry $V^{\mathrm{switch},t}$ that transforms registers $B^t$ of ${\effstate(0, n)}$ (after $t$ \Gen queries and no \Reflect queries) into the register $\hat B$ of $\hat\effstate(0,q, n)$, in such a way that the corresponding global states on $A^tB^t$ and $A^t \hat B$ are mapped to each other. The isometry is partial because its domain is the symmetric subspace of ${\C^{2^n}}^{\otimes t}$. It is defined as follows:
	\begin{equation}
		V^{\mathrm{switch},t}_{B^t\to\hat B}=\sqrt{\frac{d_{\mathrm{Sym}^t\C^d{2^n}}}{|I|}}\sum_{i\in I} \left(\bra 0 U^T_i\right)^{\otimes t}_{B^t}\otimes\ket i_{\hat B}\,.
	\end{equation}

To verify that this is indeed the desired isometry, we calculate:
	\begin{align}
		\left(\bra 0 U^T_i\right)^{\otimes t}_{B^t}\ket{\phi^+_{\mathrm{Sym}}}_{A^tB^t}=&\sqrt{\frac{2^{nt}}{d_{\mathrm{Sym}^t\C^{2^n}}}}\left(\bra 0 U^T_i\right)^{\otimes t}_{B^t}\Pi^{\mathrm{Sym}}_{B^t}\ket{\phi^+}_{A^tB^t}\\
		=&\sqrt{\frac{2^{nt}}{d_{\mathrm{Sym}^t\C^{2^n}}}}\left(\bra 0 U^T_i\right)^{\otimes t}_{B^t}\ket{\phi^+}_{A^tB^t}\\
		=&\sqrt{\frac{2^{nt}}{d_{\mathrm{Sym}^t\C^{2^n}}}}\left(\bra 0 \right)^{\otimes t}_{B^t}\otimes \left(U_i\right)^{\otimes t}_{A^t}\ket{\phi^+}_{A^tB^t}\\
		=&\sqrt{\frac{1}{d_{\mathrm{Sym}^t\C^{2^n}}}} \left(U_i\ket 0\right)^{\otimes t}_{A^t}.
	\end{align} 
	Here we have used the fact that $\left(\bra 0 U^T_i\right)^{\otimes t}$ is in the symmetric subspace in the second equality, and the third and forth equality are applications of the Mirror Lemma (\expref{Lemma}{lem:mirror}) with $d=d'=2^{nt}$, and $d=1,\ d'=2^{nt}$, respectively. 

	We have hence proven the exact correctness of $\effstate(0, n)$ without the \Reflect interface. Note that the global state after $t$ queries to $\effstate(0, n).\Gen$ is  the maximally entangled state of two copies of the $t$-fold symmetric subspace; of course, this is only true up to actions performed by the adversary, but those trivially commute with maps applied only to the oracle registers. As the global state is in the domain of $V^{\mathrm{switch},t}_{B^t\to\hat B}$, we obtain the equation
	\begin{equation}\label{eq:intertwine}
			\hat V^{t\to t+1}_{\hat B\to \hat B A_{t+1}}V^{\mathrm{switch},t}_{B^t\to\hat B}=V^{\mathrm{switch},t+1}_{B^{t+1}\to\hat B}V^{t\to t+1}_{B^t\to B^{t+1} A_{t+1}}\,.
	\end{equation}
More precisely, we observe that the two sides of the above have the same effect on the global state, and then conclude that they must be the same operator by the Choi-Jamoi\l kowski isomorphism. 

Recalling that $V^{\mathrm{switch},t}$ is partial with the symmetric subspace as its domain, we see that Equation \eqref{eq:intertwine} is equivalent to
	 \begin{align}
	 \left(V^{\mathrm{switch},t+1}_{B^{t+1}\to\hat B}\right)^\dagger\hat V^{t\to t+1}_{\hat B\to \hat B A_{t+1}}V^{\mathrm{switch},t}_{B^t\to\hat B}=&\Pi^{\Sym^{t+1}\C^{2^n}}_{B^{t+1}}V^{t\to t+1}_{ B^t\to B^{t+1} A_{t+1}}\\
	 =&V^{t\to t+1}_{ B^t\to B^{t+1} A_{t+1}}\Pi^{\Sym^{t}\C^{2^n}}_{B^{t}}\,.
	 \end{align}

By taking the above equality times its adjoint, we arrive at
	 \begin{align}
&	\left(V^{\mathrm{switch},t}_{B^t\to\hat B}\right)^\dagger\left(\hat V^{t\to t+1}_{\hat B\to \hat B A_{t+1}}\right)^\dagger V^{\mathrm{switch},t+1}_{B^{t+1}\to\hat B} \left(V^{\mathrm{switch},t+1}_{B^{t+1}\to\hat B}\right)^\dagger\hat V^{t\to t+1}_{\hat B\to \hat B A_{t+1}}V^{\mathrm{switch},t}_{B^t\to\hat B}\nonumber\\
=&\Pi^{\Sym^{t}\C^{2^n}}_{B^{t}}\left(V^{t\to t+1}_{ B^t\to B^{t+1} A_{t+1}}\right)^\dagger V^{t\to t+1}_{ B^t\to B^{t+1} A_{t+1}}\Pi^{\Sym^{t}\C^{2^n}}_{B^{t}}.
	 \end{align}
By Equation \eqref{eq:intertwine}, the range of $\hat V^{t\to t+1}_{\hat B\to \hat B A_{t+1}}V^{\mathrm{switch},t}_{B^t\to\hat B}$ is contained in the range of $V^{\mathrm{switch},t+1}_{B^{t+1}\to\hat B}\otimes \mathds 1_{A_{t+1}}$. We can thus simplify as follows:
	 \begin{align}\label{eq:switch}
	 \left(V^{\mathrm{switch},t}_{B^t\to\hat B}\right)^\dagger&\left(\hat V^{t\to t+1}_{\hat B\to \hat B A_{t+1}}\right)^\dagger \hat V^{t\to t+1}_{\hat B\to \hat B A_{t+1}}V^{\mathrm{switch},t}_{B^t\to\hat B}\nonumber\\
	 =&\Pi^{\Sym^{t}\C^{2^n}}_{B^{t}}\left(V^{t\to t+1}_{ B^t\to B^{t+1} A_{t+1}}\right)^\dagger V^{t\to t+1}_{ B^t\to B^{t+1} A_{t+1}}\Pi^{\Sym^{t}\C^{2^n}}_{B^{t}}.
	 \end{align}
Now observe that both sides of the above consist of a projection operator ``sandwiched'' by some operation. These two projection operators are precisely the projectors which define the reflection operators of $\hat\effstate(0,q,n)$ (on the left-hand side) and $\effstate(0,n)$ (on the right-hand side.) We thus see that Equation \eqref{eq:switch} shows that applying $\effstate(0,n).\Reflect$ is the same as switching to $\hat\effstate(0,q,n)$, applying $\hat\effstate(0,q,n).\Reflect$, and then switching back to $\effstate(0,n)$. The same holds for the controlled versions $\effstate(0,n).\CReflect$ and $\hat\effstate(0,n).\CReflect$.

This completes the proof of the exact equality between the stateful machines $\idealstate(n)$ and $\effstate(0,n)$. As argued at the start of the proof, the approximation case follows.
\qed \end{proof}

\section{Simulating a Haar-random unitary oracle}\label{sec:unitary}

\subsection{The problem, and our approach}

We begin by defining the ideal object we'd like to emulate. This ideal object samples a Haar-random unitary $U$, and then answers two types of queries: queries to $U$, and queries to its inverse $U^\dagger$.

\begin{construction}[Ideal unitary sampler]\label{con:unitary-ideal}
Let $n$ be a positive integer. The ideal unitary sampler is an ideal machine $\idealunitary(n)$ with interfaces $(\Init, \Eval, \Invert)$, defined as follows.
\begin{enumerate}
\item $\idealunitary(n).\Init:$ takes no input; samples a description $\tilde U$ of a Haar-random $n$-qubit unitary operator $U$.
\item $\idealunitary(n).\Eval:$ takes $n$-qubit register as input, applies $U$ and responds with the output;
\item $\idealunitary(n).\Invert:$ takes $n$-qubit register as input, applies $U^{-1}$ and responds with the output.
\end{enumerate}
\end{construction}

Below, we construct a stateful machine that runs in polynomial \emph{space} (and the runtime of which we don't characterize), and that is indistinguishable from $\idealunitary(n)$ for arbitrary query algorithms.


%

\subsubsection{Our approach.}  It turns out that the solution of a much easier task comes to our help, namely simulating a Haar random unitary for an algorithm that makes an \emph{a priori} polynomially bounded number  $t$ of queries. In this case we can just pick a unitary $t$-design, sample an element from it and answer the up to $t$ queries using this element. As in the proof of Theorem \ref{thm:state-sampler}, we can also construct an isometric stateful machine version of this strategy: Instead of sampling a random element from the $t$-design, we can prepare a quantum register in a superposition, e.g. over the index set of the $t$-design (\Init), and then apply the $t$-design element (\Eval) or its inverse (\Invert) controlled on that register.

Now consider an algorithm that makes $t$ parallel queries to a Haar random unitary (for ease of exposition let us assume here that the algorithm makes no inverse queries). The effect of these $t$ parallel queries is just the application of the $t$-twirling channel (or the mixed twirling channel defined in Equation \eqref{eq:mixed}) to the $t$ input registers. The $t$-design-based isometric stateful machine simulates this $t$-twirling channel faithfully. What is more, it applies a Stinespring dilation of the $t$-twirling channel, the dilating register being the one created by initialization.

Now suppose we have answered $t$ queries using the $t$-design-based machine, and are now asked to answer another, still parallel, query. Of course we cannot, in general, just answer it using the $t$-design, as its guarantees only hold for $t$ applications of the unitary. But all Stinespring dilations of a quantum channel are equivalent in the sense that there exists a (possibly partial) isometry acting on the dilating register of one given dilation, that transforms it into another given dilation. So we can just apply an isometry that transforms our $t$-design based Stinespring dilation into a $t+1$-design based one, and subsequently answer the $t+1$st query using a controlled unitary.

\subsection{Construction and proof}




We continue to describe a stateful machine that simulates $\idealunitary(n)$ exactly
and has a state register of size polynomial in $n$ and the total number of queries $q$ that an algorithm makes to its \Eval and \Invert interfaces. The existence of the required unitary $t$-designs is due to \expref{Corollary}{cor:udesigns}.

We recall our conventions for dealing with many copies of fixed-sized registers. We let $A$ denote an $n$-qubit register, we let $A_j$ denote indexed copies of $A$, and we let $A^t$ denote $A_1 A_2 \cdots A_t$. In this case, the various copies of $A$ will be the input registers of the adversary, on which the simulator will act. The oracle will now hold a single register $\hat B_t$ whose size will grow with the number of queries $t$. This register holds an index of an element in a $t$-design.

For the construction below, we need the following quantum states and operators. For a positive integer $n$, choose a family of $n$-qubit  unitary designs $\{D_{t}\}_{t\in\N}$, where $D_{t}=\{U_{t,i}\}_{i\in I_{t}}$ is a unitary $t$-design. Let $\hat B_t$ be a register of dimension $|I_{t}|$ and define the uniform superposition over indices
\begin{equation}
\ket{\eta_t}_{\hat B_t}=\frac{1}{\sqrt{\left|I_{t}\right|}}\sum_{i\in I_{t}}\ket i_{\hat B_t}.
\end{equation}
For nonnegative integers $t, t', \ell$, define the unitaries 
\begin{align}
V^{(t,t',\ell)}_{A^{t'}\hat B_t}&=\sum_{i\in I_{t}}\left(U_{t,i}\right)^{\otimes \ell}_{A_1A_2...A_\ell}\otimes \left(U_{t,i}^\dagger\right)^{\otimes t'-\ell}_{A_{\ell+1}A_{\ell+2}...A_{t'}}\otimes \proj i_{\hat B_t}\,.
\end{align} 
These isometries perform the following: controlled on an index $i$ of a $t$-design $U_{t,i}$, apply $U_{t, i}$ to $\ell$ registers and $U_{t,i}^\dagger$ to $t' - \ell$ registers. For us it will always be the case that $t' \leq t$, since otherwise the $t$-design property no longer makes the desired guarantees on the map $V$.

We also let $W^{(t,\ell)}_{\hat B_{t}\to \hat B_{t+1}}$ be an isometry such that
\begin{equation}
V^{(t+1,t,\ell)}_{A^t\hat B_{t+1}}\ket{\eta_{t+1}}_{\hat B_{t+1}}=W_{\hat B_{t}\to \hat B_{t+1}}V^{(t,t,\ell)}_{A^t\hat B_{t}}\ket{\eta_{t}}_{\hat B_{t}}
\end{equation}
for $\ell=0,...,t$. The isometry $W$ always exists, as all Stinespring dilations are isometrically equivalent, and both $V^{(t,t,\ell)}_{A^t\hat B_{t}}\ket{\eta_{t}}_{\hat B_{t}}$ and $V^{(t+1,t,\ell)}_{A^t\hat B_{t+1}}\ket{\eta_{t+1}}_{\hat B_{t+1}}$ are Stinespring dilations of the mixed twirling channel $\mathcal T^{(t, \ell)}$ by the $t$-design property. 

We are now ready to define the space-efficient unitary sampler.
\begin{construction}[Space-efficient unitary sampler]\label{con:unitary-sampler}
Let $n$ be a positive integer and $\{D_{t}\}_{t\in\N}$ a family of $n$-qubit unitary $t$-designs $D_{t}=\{U_{t,i}\}_{i\in I_{t}}$, with $|I_{t}|=2^{\mathrm{poly}(n,t)}$. Define a stateful machine $\effunitary(n, \epsilon)$ with interfaces $(\Init, \Eval,$ $\Invert)$ as follows. The machine will maintain counters $t_e$ (the number of $\Eval$ queries), $t_i$ (the number of $\Invert$ queries), and $t := t_e + t_i$.
\begin{enumerate}
\item $\effunitary(n).\Init:$  Prepares the state $\ket{\eta_1}_{\hat B_1}$ and stores it.
\item $\effunitary(n).\Eval:$ 
\begin{itemize}
\item If $t=0$, apply $V^{(1,1,1)}_{A_1\hat B_1}$, where $A_1$ is the input register. 
\item If $t>0$, apply $W^{(t,t_e)}_{\hat B_{t}\to \hat B_{t+1}}$ to the state register and subsequently apply $V^{t+1,1,1}_{A_{t+1}\hat B_{t+1}}$, where $A_{t+1}$ is the input register. 
\end{itemize}
\item $\idealunitary(n).\Invert:$
\begin{itemize}
\item If $t=0$, apply $V^{(1,1,0)}_{A_1\hat B_1}$, where $A_1$ is the input register. 
\item If $t>0$, apply $W^{(t,t_e)}_{\hat B_{t}\to \hat B_{t+1}}$ to the state register and subsequently apply $V^{t+1,1,0}_{A_{t+1}\hat B_{t+1}}$, where $A_{t+1}$ is the input register. 
\end{itemize}
\end{enumerate}
\end{construction}

We want to show that the above sampler is indistinguishable from the ideal sampler to any oracle algorithm, in the following sense. Given a stateful machine $\algo C \in \{\idealunitary(n), \effunitary(n, \epsilon)\}$ and a (not necessarily efficient) oracle algorithm $\algo A$, we define the process $b \from \algo A^{\algo C}$ as follows:
\begin{enumerate}
	\item $\algo C.\Init$ is called;
	\item $\algo A$ receives oracle access to $\algo C.\Eval$ and $\algo C.\Invert$;
	\item $\algo A$ outputs a bit $b$\,.
\end{enumerate}

\begin{theorem}\label{thm:unitary-sampler}
	For all oracle algorithms $\algo A$
	\begin{equation}\label{eq:unitary-thm}
	\Pr\left[\algo A^{\idealunitary(n)} = 1\right]=\Pr\left[\algo A^{\effunitary(n, \epsilon)} = 1\right].
	\end{equation}
\end{theorem}
\begin{proof}
	We begin by proving the following claim by induction. The claim states that the theorem holds for adversaries who only make parallel queries.
	\begin{claim}
		For all $x\in\{0,1\}^t$, let $V^{(x)}_{A^t\to A^t\hat B_t}$ be the isometry that is implemented by making $t$ parallel queries to $\effunitary(n,\epsilon)$, where the $i$-th query is made to the \Eval interface if $x_i=1$ and to the $\Invert$ interface if $x_i=0$. Let further $\sigma\in S_t$ be a permutation such that $\sigma.x=11...100...0$, where the lower dot denotes the natural action of $S_t$ on strings of length $t$. Then
			\begin{equation}\label{eq:par}
		V^{(x)}_{A^t\to A^t\hat B_t}=\sigma^{-1}_{A^t}V^{(t,t,\ell)}_{A^{t}\hat B_t}\ket{\eta_t}_{\hat B_t},
		\end{equation}
		where $\sigma$ acts by permuting the $t$ registers.
	\end{claim}
\begin{proof}
	For $t=1$, the claim trivially holds. Now suppose the claim holds for $t-1$. By definition of the \Eval and \Invert interfaces, 
\begin{equation}\label{eq:induc0}
	V^{(x)}_{A^t\to A^t\hat B_t}=V^{t,1,x_t}_{A_{t}\hat B_{t}}W^{(t,\ell)}_{\hat B_{t-1}\to \hat B_{t}}V^{(x_{[1;t-1]})}_{A^{t-1}\to A^{t-1}\hat B_{t-1}},
\end{equation}
where $x_{[a,b]}=x_ax_{a+1}...x_b$. 
By the induction hypothesis, we have
\begin{equation}\label{eq:induc1}
		V^{(x_{[1;t-1]})}_{A^{t-1}\to A^{t-1}\hat B_{t-1}}=\hat\sigma^{-1}_{A^{t-1}}V^{(t-1,t-1,\ell-x_t)}_{A^{t-1}\hat B_{t-1}}\ket{\eta_{t-1}}_{\hat B_{t-1}}
\end{equation}
for an appropriate permutation $\hat\sigma\in S_{t-1}$.
By the design property of $D_j$ for $j=t, t-1$ and the definition of $W^{(t,\ell)}$ we obtain
\begin{align}
&\mathcal T^{(t-1,\ell-x_t)}_{D_{t-1}}=\mathcal T^{(t-1,\ell-x_t)}_{D_{t}}\nonumber\\
\Leftrightarrow\quad &W^{(t-1,\ell)}_{\hat B_{t-1}\to \hat B_{t}}V^{(t-1,t-1,\ell-x_t)}_{A^{t-1}\hat B_{t-1}}\ket{\eta_{t-1}}_{\hat B_{t-1}}=V^{(t,t-1,\ell-x_t)}_{A^{t-1}\hat B_{t}}\ket{\eta_{t-1}}_{\hat B_{t}}\nonumber\\
\Leftrightarrow\quad &W^{(t,\ell)}_{\hat B_{t-1}\to \hat B_{t}}\hat\sigma^{-1}_{A^{t-1}}V^{(t-1,t-1,\ell-x_t)}_{A^{t-1}\hat B_{t-1}}\ket{\eta_{t-1}}_{\hat B_{t-1}}=\hat\sigma^{-1}_{A^{t-1}}V^{(t,t-1,\ell-x_t)}_{A^{t-1}\hat B_{t}}\ket{\eta_{t-1}}_{\hat B_{t}}.\label{eq:induc2}
\end{align}
Here we have used the fact that the permutation and $W^{(t-1,\ell)}$ commute because they act on disjoint sets of registers. 
Putting Equations \eqref{eq:induc0}, \eqref{eq:induc1} and \eqref{eq:induc2} together, it follows that 
\begin{align}
V^{(x)}_{A^t\to A^t\hat B_t}=V^{t,1,x_t}_{A_{t}\hat B_{t}}\hat\sigma^{-1}_{A^{t-1}}V^{(t,t-1,\ell-x_t)}_{A^{t-1}\hat B_{t}}\ket{\eta_{t}}_{\hat B_{t}}.
\end{align}
But clearly
\begin{equation}
	V^{t,1,x_t}_{A_{t}\hat B_{t}}\hat\sigma^{-1}_{A^{t-1}}V^{(t,t-1,\ell-x_t)}_{A^{t-1}\hat B_{t}}=\sigma^{-1}_{A^{t}}V^{(t,t,\ell)}_{A^{t}\hat B_{t}}
\end{equation}
For an appropriate permutation $\sigma$ that consists of applying $\hat\sigma$ and then sorting in $x_t$ correctly.
\end{proof}
The generalization to adaptive algorithms is done via \emph{post-selection}: Given an algorithm $\algo A$, consider non-adaptive algorithm $\tilde{\algo A}$ that first queries the \Eval and \Invert interfaces of the stateful machine it is interacting with on the first halves of a sufficient number of maximally entangled states. Subsequently the adaptive adversary is run, answering the queries by running quantum teleportation on the inputs together with the remaining halves of the maximally entangled states. This way, the query registers of the adaptive  queries are teleported into the previously made non-adaptive queries, but of course they incur a random Pauli error on the way, that cannot be corrected.

As the output of $\tilde{\algo A}$ is, however, \emph{exactly} the same whether it interacts with $\idealunitary(n)$ or with $\effunitary(n,0)$, the same holds for the version of $\tilde{\algo A}$ where we post-select, or condition, on the outcome that all the Pauli corrections in all the teleportation protocols are the identity. But this post-selected algorithm has the same output as $\algo A$ no matter what oracles it is given.
\qed \end{proof}


 Using \expref{Corollary}{cor:udesigns} and the above, we get the following upper bound on the space complexity of lazy sampling Haar random unitaries.
\begin{corollary}
	The space complexity $S$ of simulating $\idealunitary(n)$ as a function of $n$ and the number of queries $q$  is bounded from above by the logarithm of number of elements in any family of exact $n$-qubit unitary $q$-designs, and hence
	\begin{equation}
		S(n,q)
		\le 2q(2n+\log e)+O(\log q)\,.
	\end{equation}
\end{corollary}
\begin{proof}
	According to \expref{Corollary}{cor:udesigns}, There exists an exact unitary $q$-design such that $2q\log\left( \frac{e(2^{2n}+q-1)}{q}\right)\le 2q(2n+\log e)$ qubits suffice to coherently store the index of an element from it. The only additional information that $\effunitary(n)$ needs to store is how many direct and inverse queries have been answered, which can be done using $\log q$ bits. \flushright \qed
	
\end{proof}

Our results suggest two possible approaches to devise a time-efficient lazy sampler for Haar random unitaries. The most promising one is to use the same approach as for the state sampler and explicitly constructing the update isometry, possibly using explicit bases for the irreducible representations of $U(2^n)$, or using the Schur transform \cite{Bacon2006}. The other one would be to use the $t$-design update method described above, but using efficient approximate $t$-designs, e.g. the ones constructed in \cite{BHH16}. This would, however, likely require a generalization of the Stinespring dilation continuity result from \cite{Kretschmann2008} to so-called quantum combs \cite{Chiribella2008}. In addition, we would need to show that the transition isometries, i.e. the approximate analogue of the isometries $W^{(t,\ell)}$ from \expref{Construction}{con:unitary-sampler}, are efficiently implementable. We leave the exploration of these approaches for future work.

\section{Application: untraceable quantum money}\label{sec:money}

\subsection{Untraceable quantum money}

Our definition of quantum money deviates somewhat from others in the literature~\cite{AC12,Ji2018}. We allow the bank to maintain an internal quantum register, we do not require that the money states are pure, and we allow adversaries to apply arbitrary (i.e., not necessarily efficiently implementable) channels.
\begin{definition}[Quantum money]
A quantum money scheme is a family of stateful machines $\mathfrak M$ indexed by a security parameter $\lambda$, and having two interfaces:
\begin{enumerate}
\item $\Mint$: receives no input, outputs an $n$-qubit register;
\item $\Ver$: receives an $n$-qubit register as input, outputs an $n$-qubit register together with a flag $\{\acc, \rej\}$,
\end{enumerate}
satisfying the following two properties:
\begin{itemize}
\item correctness: $\|\Ver \circ \Mint - \one \otimes \proj\acc\| \leq \negl(\lambda)$;\footnote{Note that it is understood that this inequality should hold no matter which interfaces have been called in between the relevant \Mint and \Ver calls}
\item unforgeability: for all channels $\Lambda$ with oracle, and all $k \geq 0$, 
$$
\Pr\left[\acc^{k+1} \leftarrow {}_\mathrm{flag}|\Ver^{\otimes k+1} \circ \Lambda^\Ver \circ \Mint^{\otimes k}\right] \leq \negl(\lambda)\,,
$$
where ${}_\mathrm{flag}|$ denotes discarding all registers except $\Ver$ flags.
\end{itemize}
\end{definition}

It is implicit in the definition that $n$ is a fixed polynomial function of $\lambda$, and that all relevant algorithms are uniform in $\lambda$.

Next, we define untraceability for quantum money schemes.

\begin{definition}[Untraceability game]
The untraceability game $\Untrace_\lambda[\mathfrak M, \algo A]$ between an adversary $\algo A$ and a quantum money scheme $\mathfrak M$ at security parameter $\lambda$ proceeds as follows:
\begin{enumerate}
\item \textsf{set up the trace}: $\algo A(1^\lambda)$ receives oracle access to $\Ver$ and $\Mint$, and outputs registers $M_1$, $M_2$, \dots, $M_k$ and a permutation $\pi \in S_k$;
\item \textsf{permute and verify bills}: $b \from \{0,1\}$ is sampled, and if $b=1$ the registers $M_1 \cdots M_k$ are permuted by $\pi$. $\Ver$ is invoked on each $M_j$; the accepted registers are placed in a set $\mathcal M$ while the rest are discarded;
\item \textsf{complete the trace}: $\algo A$ receives $\mathcal M$ and the entire internal state of $\mathfrak M$, and outputs a guess $b' \in \{0,1\}$.
\end{enumerate}
The output of $\Untrace_\lambda[\mathfrak M, \algo A]$ is $\delta_{bb'}$; in the case $b=b'$, we say that $\algo A$ wins.
\end{definition}

\begin{definition}[Untraceable quantum money]
A quantum money scheme $\mathfrak M$ is untraceable if, for every algorithm $\algo A$,
$$
\Pr\left[1 \from \Untrace_\lambda[\mathfrak M, \algo A] \right] \leq \frac{1}{2} + \negl(\lambda)\,.
$$
\end{definition}

The intuition behind the definition is as follows. In general, one might consider a complicated scenario involving many honest players and many adversaries, where the goal of the adversaries is to trace the movement of at least one bill in transactions involving at least one honest player. Tracing in transactions involving only adversaries is of course trivial. The first natural simplification is to view all the adversaries as a single adversarial party; if that party cannot trace, then neither can any individual adversary. Next, we assume that honest players will verify any bills they receive immediately; obviously, if they do not do this, and then participate in transactions with the adversary, then tracing is again trivial. We thus arrive at the situation described in the game: the adversary is first allowed to create candidate bills arbitrarily, including storing information about them and entangling them with additional registers, before handing them to honest players who may or may not perform some transactions; the goal of the adversary is to decide which is the case, with the help of the bank. Note that one round of this experiment is sufficient in the security game, as an adversary can always use the $\Ver$ and $\Mint$ oracles to simulate additional rounds.

One might reasonably ask if there are even stronger definitions of untraceability than the above. Given its relationship to the ideal state sampler, we believe that Haar money, defined below, should satisfy almost any notion of untraceability, including composable notions. We also remark that, based on the structure of the state simulator, which maintains an overall pure state supported on two copies of the symmetric subspace of banknote registers, it is straightforward to see that the scheme is also secure against an ``honest but curious'' or ``specious''~\cite{SSS09,DNS10} bank. We leave the formalization of these added security guarantees to future work.

\subsection{Haar money}

Next, we show how the lazy state sampler (\expref{Construction}{con:state-sampler}) yields untraceable quantum money. The construction follows the idea of \cite{Ji2018} sample a single (pseudo)random quantum state and hand out copies of it as banknotes.

\begin{construction}[Haar money]\label{con:haar-money}
Let $n$ be a positive integer and $\epsilon > 0$. The Haar scheme $\haarmoney(n, \epsilon)$ is defined as follows:
\begin{itemize}
\item $\Mint$: on first invocation, instantiate $\effstate := \effstate(n, \epsilon)$ by running $\effstate.\Init$. On all invocations, output result of $\effstate.\Gen$;
\item $\Ver$: apply $\effstate.\Ver$; in the $\acc$ case, call $\Mint$ and output the result; in the $\rej$ case, output $0^n$.
\end{itemize}
\end{construction}

We remark that, while \expref{Construction}{con:state-sampler} does not explicitly include a $\Ver$ interface, one can easily be added by \expref{Lemma}{lem:ver-reflect}.

\begin{proposition}
Haar money is an untraceable quantum money scheme.
\end{proposition}
\begin{proof}
We need to show three properties: completeness, unforgeability, and untraceability. For the completeness and unforgeability properties, observe that \expref{Theorem}{thm:state-sampler} implies that the adversary's view is indistinguishable (up to negligible terms) if we replace the efficient state sampler $\effstate$ with the ideal $\idealstate$. Once we've made that replacement, completeness follows from the definition of $\idealstate.\Gen$ and $\idealstate.\Ver$, and unforgeability follows from the complexity-theoretic no-cloning theorem~\cite{AC12}.

For untraceability, it is of course true that $\idealstate$ is obviously untraceable. However, we cannot simply invoke \expref{Theorem}{thm:state-sampler} to conclude the same about $\effstate$, since the adversary will receive the state of the bank at the end of the game. Instead, we argue as follows. Consider step $2$ (permute and verify bills) in the untraceability game $\Untrace_\lambda[\haarmoney, \algo A]$. An equivalent way to perform this step is to (i.) verify all the registers first, (ii.) discard the ones that fail verification, and then (iii.) apply the permutation, conditioned on the challenge bit $b$. Steps (i.) and (ii.) are applied always and in particular do not depend on $b$. However, after (i.) and (ii.) have been applied, by the definition of $\effstate$ the joint state of the bank and all the $M_j \in \mathcal M$ (and indeed all verified bills in existence) is negligibly far from the state $\ket{\phi^+_\Sym}$, i.e., the maximally entangled state on the symmetric subspace. This state is clearly invariant under permutation of the money registers, and in particular under the permutation of the registers in $\mathcal M$ selected by the adversary. We emphasize that this invariance holds for the entire state (including the bank.) As the remainder of the game experiment is simply some channel applied to that state, and this channel does not depend on $b$, the result follows.
\qed \end{proof}

While Haar money is an information-theoretically unforgeable and untraceable quantum money scheme, it is easy to see that the quantum money scheme devised in \cite{Ji2018} is \emph{computationally} unforgeable and untraceable.







\appendix

\section{State preparation lemma proofs}\label{sec:prep-lemmas}

We now prove the state preparation lemmas from the preliminaries.

\begin{lemma}[Restatement of \expref{Lemma}{lem:poly-prepare}]
Let $\ket{\varphi} = \sum_{x \in \{0,1\}^n} \varphi(x) \ket{x}$ be a family of quantum states whose amplitudes $\varphi$ have an efficient classical description $\tilde \varphi$, and such that $|\{x : \varphi(x) \neq 0\}| \leq \poly$. Then there exists a quantum algorithm $\mathcal P$ which runs in time polynomial in $n$ and $\log(1/\epsilon)$ and satisfies
$$
\|\mathcal P \ket{\tilde \varphi}\ket{0^n} - \ket{\tilde \varphi}\ket{\varphi}\|_2 \leq \epsilon\,.
$$
\end{lemma}
\begin{proof}
Let $\supp(\varphi) := \{x : \varphi(x) \neq 0\}$, let $t = |\supp(\varphi)|$ and let $x_1^\varphi, x_2^\varphi, \dots, x_t^\varphi$ be an indexing of $\supp(\varphi)$, e.g., by lexicographic order. We first observe that, from the classical description $\tilde \varphi$, we can efficiently compute a circuit for a $\lceil\log t\rceil$-qubit  unitary $\tilde U^{\varphi}$, such that $\tilde U^{\varphi}_{i,1}= \varphi(x_i^\varphi)$. In \cite{NC00}, Chapter 4, it is described how any $r$-qubit unitary can be implemented up to precision $\delta$ using a quantum circuit of length $O(r^2 4^r\log^c\left(r^2 4^r/\delta\right))$ for some universal constant $c$, and how to compute such a circuit efficiently. The recipe consists of a decomposition into a circuit of CNOT and arbitrary single qubit gates, and an application of the Solovay-Kitaev theorem to implement the single-qubit gates using, say, the Clifford+T gate set. The former is easily verified to be efficiently computable, and an algorithmic version of the latter can be found in, e.g., \cite{Dawson2005}. We apply this algorithm to compute a circuit for $\tilde U^\varphi$ from $\tilde \varphi$; let $m$ denote its maximum length.

The total algorithm that, on input $\tilde \varphi$, produces the circuit $C_\varphi$ for $\tilde U^{(\varphi)}$, can be written as a reversible circuit and implemented as a quantum circuit $C_{\mathrm{meta}}$ such that $U_{C_{\mathrm{meta}}}\ket{\tilde \varphi} \ket{0^m}=\ket{\tilde \varphi} \ket{C_\varphi}$. Here, $U_C$ is the unitary implemented by a quantum circuit $C$. We can then apply a universal quantum circuit to apply $C_\varphi$ to a fresh ancilla register initialized in the state $0^{\lceil\log t\rceil}$, and then again apply $U_{C_{\mathrm{meta}}}$ to uncompute the circuit description.

We are now ready to define the algorithm $\algo P$ as follows. On input $\ket{\tilde \varphi}\ket{0^n}$, we attach ancillas in the $0$ state and use the above to apply 
$$
\ket{\tilde\varphi}\ket{0^{\lceil\log t\rceil}}\ket{0^n}
\longmapsto \ket{\tilde\varphi}\tilde U\ket{0^{\lceil\log t\rceil}}\ket{0^n}
= \ket{\tilde\varphi} \sum_{i=1}^t \varphi(x_i^\varphi) \ket i \ket{0^n}\,.
$$
Next, controlled on the first two registers being in state $\ket{\tilde\varphi}\ket i$ we apply $X^{x_i^\varphi}$ to the last register. Finally, controlled on the first and last registers being in state $\ket{\tilde\varphi}$ and $\ket x$, respectively, we apply $X^i$ to the middle register if $x=x_i^\varphi$ for some $i$ (and the identity otherwise), and discard the middle register.
\end{proof}

\begin{lemma}[Restatement of \expref{Lemma}{lem:almost-uniform-prepare}]
Let $S \subset \{0,1\}^n$ be a family of sets of size $\poly$ with efficient description $\tilde S$, and let $\epsilon > 0$. There exists a quantum algorithm $\mathcal P$ which runs in time polynomial in $n$ and $\log(1/\epsilon)$ and satisfies
$$
\left\|\mathcal P \ket{\tilde S}_A\ket{0^n}_B - \ket{\tilde S}_A\ket{\bar S}_B\right\|_2 \leq \epsilon\,.
$$
\end{lemma}
\begin{proof}
We first observe that there is a quantum algorithm of size polynomial in $n$ and $\log(1/\epsilon)$ for the task of, given a classical description $\tilde S$ of a $\poly$-size set $S$, preparing the state $\ket{\bar S}$ with precision $\epsilon$. This algorithm proceeds by repeatedly preparing the uniform superposition and then applying the two-outcome measurement defined by the projector $\sum_{x \in S} \proj{x}$. After 
$$
r=\max\left(1, \left\lceil\log\left(\frac{3|S|}{\epsilon}\right)-n\right\rceil\right)
$$
 repetitions, one of the attempts will succeed with probability at least $1-\epsilon/3$. Finally, the algorithm swaps the successful register into a fixed output register (or outputs some fixed state if all attempts failed.) Let $C_\epsilon$ be the quantum circuit for executing this entire algorithm, including measurements and conditional operations.

The algorithm $\mathcal P$ will perform $C_\epsilon$ in a coherent (i.e., measurement-free) way, while uncomputing some garbage on the fly. In the end, the remaining garbage is uncomputed. The $i$-th coherent iteration step is done as follows. Initialize a qubit $C_i$ and an n-qubit register $D_i$, both in the all-zero state. Now, controlled on $C_{i-1}$ and using the convention that $C_0=1$, apply $H^n$ to $B$ and coherently measure whether $B\in S$, storing the outcome in $C_i$. Now, controlled on $C_i$, swap $B$ and $D$, and unprepare $\ket S$ in register $D_i$ with precision $\frac{\epsilon}{3r}$ using the algorithm form Lemma \ref{lem:poly-prepare}. After this procedure, $D_i$ is in the zero state and can be safely discarded. After $r$ iterations, the state is $2\epsilon/3$-close to 
\begin{equation}
	\ket{\tilde S}_A\ket{\bar S}_B\otimes \left[\sum_{\ell=0}^r\left(\frac{s}{2^n}\right)^{\ell/2}\left(\frac{2^n-s}{2^n}\right)^{1/2}\ket 1^{\otimes \ell}\otimes \ket 0^{\otimes r-\ell}\right]_{C_1C_2...C_r}.
\end{equation}
But the state the $C$-registers is in is a superposition of $r$ many computational basis states, so we can unprepare it using the algorithm from Lemma \ref{lem:poly-prepare} with precision $\epsilon/3$, so we have prepared $\ket{\tilde S}_A\ket{\bar S}_B$ up to error $\epsilon$.
\end{proof}


\begin{lemma}[Restatement of \expref{Lemma}{lem:super-prepare}]
Let $\ket{\zeta_{0,j}}, \ket{\zeta_{1,j}}$ be two familes of $n$-qubit quantum states such that $\braket{\zeta_{0,j}}{\zeta_{1,j}}=0$ for all $j$, and such that there exists a quantum algorithm $\algo P_b$ which runs in time polynomial in $n$ and $\log(1/\epsilon)$ and satisfies 
$
\|\mathcal P_b \ket{j}\ket{0^n} - \ket{j}\ket{\zeta_{b, j}}\|_2 \leq \epsilon
$
for $b \in \{0,1\}$. 

For $z_0, z_1\in\C$ such that $|z_0|^2+|z_1|^2=1$, let $\tilde z$ denote a classical description of $(z_0,z_1)$ to precision at least $\epsilon$. There exists a quantum algorithm $\algo Q$ which runs in time polynomial in $n$ and $\log(1/\epsilon)$ and satisfies
	\begin{equation}
		\left\|\algo Q \ket{j}\ket{\tilde z}\ket{0^n} - \ket{j}\ket{\tilde z}\bigl(z_0\ket{\zeta_{0,j}} + z_1\ket{\zeta_{1,j}}\bigr)\right\|_2 \leq \epsilon\,.
	\end{equation}
\end{lemma}
\begin{proof}
We first use \expref{Lemma}{lem:poly-prepare} to implement a unitary $U$ such that $U\ket{\tilde z}\ket 0=\ket{\tilde z}(z_0\ket 0+z_1\ket 1)$ up to error $\epsilon/5$. After attaching an ancillary qubit and applying this circuit, our total state is $\ket{j}\ket{\tilde z}\ket{0^n}_O (z_0\ket 0+z_1\ket 1)_Q$, where we have named some of the registers for easier reference. Note that any efficient quantum circuit has an efficient controlled version. We can thus next prepare (in register $O$) the state $\ket{\zeta_{0,j}}$ controlled on $Q$ being in state $\ket 0$, and subsequently prepare (also in $O$) the state $\ket{\zeta_{1,j}}$ controlled on $Q$ being in state $\ket 1$, both with accuracy $\epsilon/5$. Now we apply the inverse of the circuit for the preparation of $\ket{\zeta_{0,j}}$ to register $O$, without control. Controlled on $O$ being in state $\ket 0^n$, we then apply $X$ to $Q$, after which the preparation circuit for $\ket{\zeta_{0,j}}$ is applied to $O$ again (note that orthogonality of the two state families is crucial for this step.) The register $Q$ is now in the state $\ket 1$ and can be safely discarded.
\end{proof}

\end{document}